\documentclass[12pt,draftclsnofoot, onecolumn]{IEEEtran}
\usepackage{amsmath,subfigure,multirow}
\usepackage{graphicx,amssymb,lineno,bm,subfigure}
\usepackage{algorithm,color,makecell}
\usepackage{algorithmic}
\usepackage{cite}
\usepackage{array}
\usepackage{float}
\usepackage{subeqnarray}
\usepackage{cases}
\usepackage{url}

\newtheorem{thm}{\textbf{Theorem}}
\newtheorem{rmk}{\textbf{Remark}}
\newtheorem{lma}{\textbf{Lemma}}

\newtheorem{prop}{\textbf{Proposition}}
\newtheorem{corol}{\textbf{Corollary}}

\newcommand{\tabincell}[2]{\begin{tabular}{@{}#1@{}}#2\end{tabular}}
\IEEEoverridecommandlockouts
\begin{document}

%
\title{\huge{Stochastic Joint Radio and Computational Resource Management for Multi-User Mobile-Edge Computing Systems}}
\author{\IEEEauthorblockN{Yuyi Mao, Jun Zhang, S.H. Song, and Khaled B. Letaief, \emph{Fellow, IEEE}}


\thanks{The authors are with the Department of Electronic and Computer Engineering, The Hong Kong University
of Science and Technology, Clear Water Bay, Kowloon, Hong Kong (e-mail: \{ymaoac, eejzhang, eeshsong, eekhaled\}@ust.hk). Khaled B. Letaief is also with Hamad bin Khalifa University, Doha, Qatar (e-mail: kletaief@hbku.edu.qa).
}
\thanks{Part of this work was presented at IEEE Global Communications Conference (GLOBECOM), Washington, DC, Dec. 2016 \cite{YMAO1612a}.}
}

\maketitle

\vspace{-10pt}
\begin{abstract}
Mobile-edge computing (MEC) has recently emerged as a prominent technology to liberate mobile devices from computationally intensive workloads, by offloading them to the proximate MEC server. To make offloading effective, the radio and computational resources need to be dynamically managed, to cope with the time-varying computation demands and wireless fading channels. In this paper, we develop an online joint radio and computational resource management algorithm for multi-user MEC systems, with the objective as minimizing the long-term average weighted sum power consumption of the mobile devices and the MEC server, subject to a task buffer stability constraint. Specifically, at each time slot, the optimal CPU-cycle frequencies of the mobile devices are obtained in closed forms, and the optimal transmit power and bandwidth allocation for computation offloading are determined with the \emph{Gauss-Seidel method}; while for the MEC server, both the optimal frequencies of the CPU cores and the optimal MEC server scheduling decision are derived in closed forms. Besides, a delay-improved mechanism is proposed to reduce the execution delay. Rigorous performance analysis is conducted for the proposed algorithm and its delay-improved version, indicating that the weighted sum power consumption and execution delay obey an $\left[O\left(1\slash V\right),O\left(V\right)\right]$ tradeoff with $V$ as a control parameter. Simulation results are provided to validate the theoretical analysis and demonstrate the impacts of various parameters.
\end{abstract}
\vspace{-10pt}
%
\IEEEpeerreviewmaketitle
\section{Introduction}
The increasing popularity of smart mobile devices is driving the development of computation-intensive mobile applications with advanced features, e.g., interactive online gaming, gesture and face recognition, voice control, as well as 3D modeling. This poses more stringent requirements on the quality of computation experience, which cannot be easily satisfied by mobile devices due to their limited resources, e.g., the processing speed, memory size, and battery energy. As a result, new solutions to handle the explosive computation demands and the ever-increasing computation quality requirements are emerging \cite{Gubbi1309}. \emph{Mobile-edge computing} (MEC) is such a promising technology to release the tension between the computation-intensive applications and the resource-limited mobile devices \cite{ETSI14,Barbarossa1411,YMao17MECSurvey}. Different from conventional cloud computing systems, which rely on remote public clouds that will induce long latency due to data exchange, MEC offers computation capability within the radio access network. Therefore, by offloading the computation tasks from the mobile devices to the MEC servers, the quality of computation experience, including energy consumption and execution latency, can be greatly improved \cite{Satyanarayanan0910,Kumar1004,Kumar1302,WShi16}.

\subsection{Related Works}
Computation offloading for cloud computing systems has attracted significant attention from computer science and communications research communities in recent years. In order to prolong the battery lifetime and improve the computation performance, various code offloading frameworks, e.g., MAUI \cite{Cuervo1006} and ThinkAir \cite{Kosta1203}, were proposed. Nevertheless, the efficiency of computation offloading for MEC highly depends on the wireless channel condition, as computation offloading requires effective wireless data transmission between mobile devices and MEC servers. Therefore, computation offloading policies that take the wireless channel condition into consideration have been extensively studied most recently \cite{WZhang1309,Munoz1510,KWang16pp,CYou1605,YYu1612,XChen1504,XChen1610,YMao16,CYou1605JSAC,LPu1612,DHuang1206,Liu1607,ZJiang1512,Kwak1512,Kim15}. In \cite{WZhang1309}, for tasks with a strict execution deadline, the local execution energy consumption was minimized by adopting \emph{dynamic voltage and frequency scaling} (DVFS) techniques, and the energy consumption for computation offloading was optimized via data transmission scheduling. In \cite{Munoz1510}, a joint optimization of communication and computational resource allocation for femto-cloud computing systems was proposed, where the cloud server is formed by a set of femto access points. This study was extended to the \emph{cloud radio access networks} (C-RANs) with mobile cloud computing capability in \cite{KWang16pp}. Besides, resource allocation policies were proposed for MEC systems based on \emph{time division multiple access} (TDMA) and \emph{orthogonal frequency-division multiple access} (OFDMA) in \cite{YYu1612,CYou1605}, {while game-theoretic decentralized computation offloading algorithms were proposed for multi-user MEC systems in \cite{XChen1504} and \cite{XChen1610} for single- and multi-channel wireless environments, respectively.} Moreover, dynamic computation offloading policies have been developed for MEC systems powered by energy harvesting \cite{YMao16} and wireless power transfer \cite{CYou1605JSAC}, for applications where replacing/recharing the device batteries is costly and difficult. {Most recently, a novel task offloading framework based on network-assisted \emph{device-to-device} (D2D) communications was proposed in \cite{LPu1612}, which enables resource sharing among the mobile users.}

However, there are some limitations in the commonly adopted assumptions in \cite{WZhang1309,Munoz1510,KWang16pp,CYou1605,YYu1612,YMao16,CYou1605JSAC,LPu1612}: It is typically assumed that the computation tasks have strict delay requirements, and no new task will be generated before the old tasks are completed\slash abandoned. Such assumptions make the computation offloading design more tractable, as only short-term performance, e.g., the performance for executing a single task for each mobile device, needs to be considered and thus the associated optimization problems are typically deterministic. Nevertheless, they may be impractical for applications that can tolerate a certain period of execution latency, e.g., multi-media streaming and file backup. For such types of applications, the long-term system performance is more relevant, and stochastic task models should be adopted. In particular, the coupling among the randomly arrived tasks cannot be ignored, and stochastic computation offloading policies should be developed \cite{DHuang1206,Liu1607,ZJiang1512,Kwak1512,Kim15}. In order to minimize the long-term average energy consumption, a stochastic control algorithm was proposed in \cite{DHuang1206}, which determines the offloaded software components in an application. In \cite{Liu1607}, a delay-optimal stochastic task scheduling policy for single-user MEC systems was proposed based on the Markov decision process. The energy-delay tradeoff in single-user MEC systems with a multi-core mobile device and heterogeneous types of mobile applications were investigated in \cite{ZJiang1512} and \cite{Kwak1512}, respectively. For MEC systems with multiple devices, the optimal design becomes more challenging compared to single-user MEC systems, as the computational resource for task execution and the radio resource for computation offloading are shared by multiple mobile devices. In other words, the optimal system operations in multi-user MEC systems are not only temporally correlated due to the random computation task arrivals, but also spatially coupled due to the competition among multiple devices. {Moreover, the freedom of parallel local and remote processing makes the mobile execution strategy, computation offloading policy, and the operations at the MEC server interdependent.} Consequently, intelligent joint allocation of the radio and computational resource should be considered to maximize the benefits of MEC. An initial investigation for multi-user MEC systems with delay-tolerant applications was conducted in \cite{Kim15}, which, however, only focused on computational resource scheduling and failed to address radio resource management.

\subsection{Contributions}
In this paper, we investigate stochastic joint radio and computational resource management for multi-user MEC systems. Our major contributions are summarized as follows:

\begin{itemize}
\item We consider a general MEC system with multiple mobile devices and a physically proximate MEC server with \emph{frequency division multiple access} (FDMA). The MEC server has limited computation capability, which
generalizes our previous work in \cite{YMAO1612a}, where the MEC server is assumed to be computationally powerful with unlimited computational resources.
\item The average weighted sum power consumption of the mobile devices and the MEC server is adopted as the performance metric, which is able to address the cost of power consumption at different nodes in MEC systems. The available radio and computational resources are jointly managed to optimize the MEC system, including the CPU-cycle frequencies for the mobile and server CPUs, the transmit power and bandwidth allocation for computation offloading, as well as the task scheduling decision at the MEC server\footnote{In the conference version of this paper \cite{YMAO1612a}, only the CPU-cycle frequencies for the local CPUs, the transmit power and bandwidth allocation for computation offloading were optimized since the MEC server is assumed to have unlimited computational resources. Besides, the power consumption of the MEC server was viewed as a constant and not included in the optimization in \cite{YMAO1612a}.}. {This is a critical but highly non-trivial design consideration for multi-user MEC systems, since the optimal system operations are temporally and spatially correlated due to the stochastic computation task arrivals and the competition among multiple devices for the available resources, respectively.}
\item An average weighted sum power consumption minimization problem subject to a task buffer stability constraint is formulated, assuming causal \emph{side information} (SI) of the task arrival and wireless channel processes. {This is a very challenging stochastic optimization problem, which involves a large amount of SI as well as decision variables.} A low-complexity online algorithm is then proposed based on Lyapunov optimization. In each time slot, the system operation is determined by solving a deterministic problem, where the CPU-cycle frequencies and the MEC server scheduling decision are obtained in closed forms, while the transmit power and bandwidth allocation are obtained through an efficient Gauss-Seidel method. Besides, a delay-improved mechanism is designed for the proposed algorithm.
\item Performance analysis is conducted for the proposed algorithm and its delay-improved version, which not only shows their capability in achieving asymptotic optimality, but also explicitly characterizes the tradeoff between the weighted sum power consumption and the execution delay. Simulation results corroborate the theoretical analysis and show that the proposed algorithms are able to balance the weighted sum power consumption and execution delay performance. In addition, the impacts of various parameters are revealed, which demonstrate the necessity of a joint consideration on radio and computational resource management for multi-user MEC systems, and offer valuable guidelines for real deployment.
\end{itemize}

\subsection{Organization}
The organization of this paper is as follows. We introduce the system model and formulate the average weighted sum power consumption minimization problem in Section II and Section III, respectively. An online joint radio and computational resource management algorithm as well as a delay-improved mechanism are developed in Section IV. We conduct performance analysis for the proposed algorithms in Section V. Simulation results will be shown in Section VI, and we will conclude this paper in Section VII.

\section{System Model}
\vspace{-20pt}
\begin{figure}[h]
\centering
\includegraphics[width=0.6\textwidth]{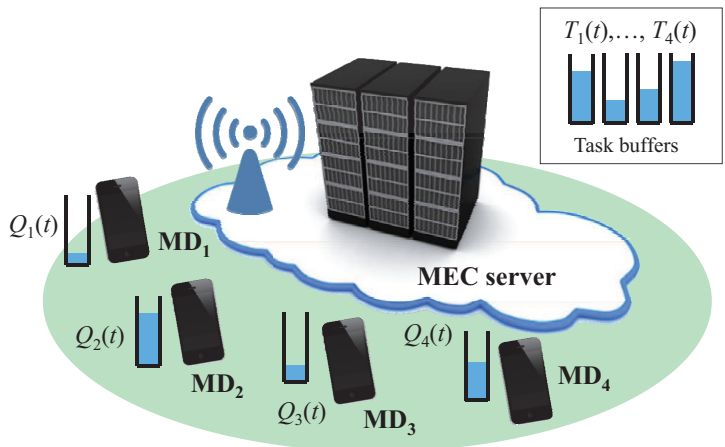}
\vspace{-10pt}
\caption{A mobile-edge computing system with four mobile devices (MDs). {The queue lengths of the task buffer at the $i$th mobile device, and the task buffer maintained by the MEC server for this device at the beginning of the $t$th time slot, are denoted as $Q_{i}\left(t\right)$ and $T_{i}\left(t\right)$, respectively.}}
\label{sysmodelMEC}
\end{figure}

We consider a mobile-edge computing (MEC) system as shown in Fig. \ref{sysmodelMEC}, where $N$ single-core mobile devices running computation-intensive applications are assisted by an MEC server. This corresponds to the scenarios where the mobile devices have relatively weak computation capability and the MEC is of supreme need, such as the \emph{Internet of Things} (IoT) applications and the \emph{wireless sensor networks} (WSNs) for surveillance \cite{Bonomi1208,CZhu1507}. However, our proposed algorithms can be adapted for multi-core mobile devices with minor modifications. The MEC server could be a small data center installed at a wireless \emph{access point} (AP) deployed by the telecom operator. Therefore, it can be accessed by the mobile devices through wireless channels, and will execute the computation tasks on behalf of the mobile devices \cite{YMao17MECSurvey,Satyanarayanan0910,WZhang1309}. {Besides, with the readily available wireless channel state information at the wireless AP, joint radio and computational resource management can be realized in MEC systems, which is an innovative feature that cannot be supported in traditional cloud computing systems.} By offloading part of the computation tasks to the MEC server, the mobile devices could not only enjoy a higher level of quality of computation experience, but also reduce the battery energy consumption \cite{Satyanarayanan0910,Kumar1004,Kumar1302,WShi16}.

The available system bandwidth is $\omega$ Hz, which is shared by the mobile devices using FDMA, and the noise power spectral density at the receiver of the MEC server is denoted as $N_{0}$. Time is slotted and the time slot length is $\tau$. For convenience, we denote the index sets of the mobile devices and the time slots as $\mathcal{N}\triangleq\{1,\cdots,N\}$ and $\mathcal{T}\triangleq\{0,1,\cdots\}$, respectively. {For ease of
reference, we list the key notations of our system model in Table \ref{notationtable}.}

{
\begin{table}[ht]
\center \protect
\caption{Summary of Key Notations}
\begin{tabular}{ll}
\Xhline{1.5pt}
{\textbf{Notation}} & {\textbf{Description}}  \tabularnewline
\Xhline{1.5pt}
{$\omega$} &System bandwidth \tabularnewline \hline
{$\mathcal{T}$ ($\mathcal{N}$)} &Index set of the time slots (mobile devices) \tabularnewline\hline
{$A_{i}\left(t\right)$} &{Amount of tasks arrived at the $i$th mobile device at the beginning time slot $t$} \tabularnewline \hline
{$Q_{i}\left(t\right)$} &{Queue length of the task buffer at the $i$th mobile device at the beginning of time slot $t$}   \tabularnewline \hline
{$T_{i}\left(t\right)$} &\tabincell{l}{Queue length of the task buffer maintained by the MEC server for the $i$th mobile \\ device at the beginning of time slot $t$} \tabularnewline \hline
{$D_{l,i}\left(t\right)$} &Amount of tasks executed locally at the $i$th mobile device in time slot $t$ \tabularnewline \hline
{$D_{r,i}\left(t\right)$} &Amount of tasks offloaded to the MEC server in time slot $t$ by the $i$th mobile device\tabularnewline \hline
{$D_{s,i}\left(t\right)$} &Amount of tasks from the $i$th mobile device executed by the MEC server in time slot $t$  \tabularnewline \hline
{$f_{i}\left(t\right)$} &CPU-cycle frequency of the $i$th mobile device in time slot $t$ \tabularnewline \hline
{$p_{l,i}\left(t\right)$} &Power consumption of the CPU at the $i$th mobile device in time slot $t$ \tabularnewline \hline
{$p_{{\rm{tx}},i}\left(t\right)$ } &Transmit power of the $i$th mobile device in time slot $t$ \tabularnewline \hline
{$f_{C,m}\left(t\right)$ } &CPU-cycle frequency of the $m$th CPU core at the MEC server in time slot $t$ \tabularnewline \hline
{$p_{\rm{ser}}\left(t\right)$ } &Power consumption of the MEC server in time slot $t$ \tabularnewline \hline
{$\alpha_{i}\left(t\right)$} &Proportion of bandwidth allocated for the $i$th mobile device in time slot $t$ \tabularnewline \hline
{$\Gamma_{i}\left(t\right)$ } &Channel power gain from the $i$th mobile device to the MEC server in time slot $t$ \tabularnewline
\Xhline{1.5pt}
\end{tabular}
\label{notationtable}
\end{table}
}

\subsection{Computation Task And Task Queueing Models}
We assume that the mobile devices are running independent and fine-grained tasks \cite{CYou1605,Kwak1512,Kim15,LChen11}: At the beginning of the $t$th time slot, $A_{i}\left(t\right)$ (bits) of computation tasks arrive at the $i$th mobile device, which can be processed starting from the $\left(t+1\right)$th time slot. Without loss of generality, we assume the $A_{i}\left(t\right)$'s in different time slots are independent and identically distributed (i.i.d.) within $\left[A_{i,\min},A_{i,\max}\right]$ with $\mathbb{E}\left[A_{i}\left(t\right)\right]=\lambda_{i},i\in\mathcal{N}$.

In each time slot, part of the computation tasks at the $i$th mobile device, denoted as $D_{l,i}\left(t\right)$, will be executed at the local CPU, while $D_{r,i}\left(t\right)$ bits of the computation tasks will be offloaded to the MEC server. {The arrived but not yet executed (or offloaded) tasks will be queued in the task buffer at each mobile device with sufficiently large capacity \cite{DHuang1206,Kwak1512,ZJiang1512,Kim15}}, and the queue lengths of the task buffers at the beginning of the $t$th time slot are denoted as $\mathbf{Q}\left(t\right)\triangleq\left[Q_{1}\left(t\right),\cdots,Q_{N}\left(t\right)\right]$, where $Q_{i}\left(t\right)$ evolves according to the following equation:
\begin{equation}
Q_{i}\left(t+1\right)=\max\{Q_{i}\left(t\right)-D_{\Sigma,i}\left(t\right),0\}+A_{i}\left(t\right),t\in \mathcal{T}.
\label{bufferdynamics}
\end{equation}
In (\ref{bufferdynamics}), $D_{\Sigma,i}\left(t\right)\triangleq D_{l,i}\left(t\right)+D_{r,i}\left(t\right)$ is the amount of tasks departing from the task buffer at the $i$th mobile device in the $t$th time slot.

{The MEC server maintains a task buffer for each mobile device to store the tasks that have been offloaded but not yet executed by the server, which is assumed to with sufficiently large capacity similar to the one at the mobile side.} Denote the queue lengths of the task buffers at the MEC server at the beginning of the $t$th time slot as $\mathbf{T}\left(t\right)\triangleq \left[T_{1}\left(t\right),\cdots,T_{N}\left(t\right)\right]$. We further denote the task scheduling decision of the MEC server at the $t$th time slot as $D_{s,n}\left(t\right),n\in\mathcal{N}$, where $D_{s,n}\left(t\right)$ is the amount of tasks from the $n$th mobile device executed by the MEC server in time slot $t$. Therefore, $T_{i}\left(t\right)$ evolves according to the following equation:
\begin{equation}
T_{i}\left(t+1\right)=\max\{T_{i}\left(t\right)-D_{s,i}\left(t\right),0\}+\min\{\max\{Q_{i}\left(t\right)-D_{l,i}\left(t\right),0\},D_{r,i}\left(t\right)\},t\in\mathcal{T},
\label{bufferdynamicsMEC}
\end{equation}
which indicates that only the tasks that have not been executed locally at the mobile devices will be stored in the task buffers at the MEC server. It is worthwhile to note that the departure function $D_{\Sigma,i}\left(t\right)$ may be larger than the amount of tasks in the corresponding local task buffer $Q_{i}\left(t\right)$, depending on the decided system operation. In this case, the excessive communication rates are allocated for transmitting dummy task inputs, which are not necessary to be processed by the MEC server. Without loss of generality, we assume the tasks buffers are empty initially, i.e., $Q_{i}\left(0\right)=T_{i}\left(0\right)=0,i\in\mathcal{N}$.

\subsection{Local Execution Model}
In order to process one bit of computation task input from the $i$th mobile device, $L_{i}$ CPU cycles will be needed, which depends on the types of applications and can be obtained by off-line measurements \cite{Miettinen10}. Denote the CPU-cycle frequency of the $i$th mobile device in the $t$th time slot as $f_{i}\left(t\right)$, which cannot exceed its maximum value $f_{i,\max}$. Thus, $D_{l,i}\left(t\right)$ can be expressed as
\begin{equation}
D_{l,i}\left(t\right)=\tau f_{i}\left(t\right)L_{i}^{-1}.
\label{Dlocal}
\end{equation}
{According to circuit theories, the CPU power is dominated by the dynamic power, which originates from the toggling activities of the logic
gates inside the CPU, and proportional to $v_{\rm{cir}}^{2}f_{c}$ in CMOS circuits, where $v_{\rm{cir}}$ and $f_{c}$ are the
circuit voltage and CPU-cycle frequency, respectively \cite{Vogeleer1309,Burd9608}. Besides, when operating at the
low voltage limits, the CPU-cycle frequency is approximately linear to the chip voltage \cite{Burd9608}}. Therefore, the power consumption for local execution at the $i$th mobile device is given by\footnote{{In order to apply the proposed algorithms in scenarios where multi-core CPUs are available at the devices, $D_{l,i}\left(t\right)$ and $p_{l,i}\left(t\right)$ should be modified as $\tau L_{i}^{-1}\cdot\sum_{z=1}^{Z_{i}}f_{i,z}\left(t\right)$ and $\sum_{z=1}^{Z_{i}}\kappa_{{\rm{mob}},i,z}f_{i,z}^{3}\left(t\right)$, respectively. Here, $Z_{i}\geq 1$ is the number of CPU cores at the $i$th mobile device,  $f_{i,z}\left(t\right)$, $\kappa_{{\rm{mob}},i,z}$ and $f_{i,z,\max}$ denote the CPU-cycle frequency, effective switched capacitance and maximum CPU-cycle frequency of the $z$th CPU core at the $i$th mobile device, respectively.  Then, $\{f_{i,z}\left(t\right)\}$'s could be determined for each mobile device under constraints $0\leq f_{i,z}\left(t\right)\leq f_{i,z,\max}, z=1,\cdots,Z_{i}, i\in\mathcal{N}$ with the proposed algorithms.}}
\begin{equation}
p_{l,i}\left(t\right)=\kappa_{{\rm{mob}},i} f^{3}_{i}\left(t\right),
\label{Plocal}
\end{equation}
where $\kappa_{{\rm{mob}},i}$ is the effective switched capacitance of the CPU at the $i$th mobile device, and it is related to the chip architecture \cite{Burd9608}.

\subsection{MEC Server Execution Model}

\subsubsection{Computation Offloading}
To offload the computation tasks for MEC server execution, the input bits of the tasks need to be delivered to the MEC server. We assume the wireless channels between the mobile devices and the MEC server are i.i.d. frequency-flat block fading. Denote the small-scale fading channel power gain from the $i$th mobile device to the MEC server in the $t$th time slot as $\gamma_{i}\left(t\right)$, which is assumed to have a bounded mean value, i.e., $\mathbb{E}\left[\gamma_{i}\left(t\right)\right]\triangleq \overline{\gamma_{i}}< \infty$. Thus, the channel power gain from the $i$th mobile device to the MEC server can be represented by $\Gamma_{i}\left(t\right)=\gamma_{i}\left(t\right)g_{0}\left(d_{0}\slash d_{i}\right)^{\theta}$, where $g_{0}$ is the path-loss constant, $\theta$ is the path-loss exponent, $d_{0}$ is the reference distance, and $d_{i}$ is the distance from the $i$th mobile device to the MEC server. {Since FDMA is utilized, according to the Shannon-Hartley formula \cite{TMCover1991}, the amount of computation tasks offloaded from the $i$th mobile device in time slot $t$ is given by
\begin{equation}
D_{r,i}\left(t\right)=
\begin{cases}
\alpha_{i}\left(t\right)\omega\tau\log_{2}\left(1+\frac{\Gamma_{i}\left(t\right)p_{{\rm{tx}},i}\left(t\right)}{\alpha_{i}\left(t\right)N_{0}\omega}\right), &\alpha_{i}\left(t\right)>0\\
0, &\alpha_{i}\left(t\right)=0,
\end{cases}
\label{Dremote}
\end{equation}
where $p_{{\rm{tx}},i}\left(t\right)$ is the transmit power with the maximum value of $p_{i,\max}$, and $\alpha_{i}\left(t\right)$ is the proportion of bandwidth allocated to the $i$th mobile device.} Denote $\bm{\alpha}\left(t\right)\triangleq\left[\alpha_{1}\left(t\right),\cdots,\alpha_{N}\left(t\right)\right]$ as the bandwidth allocation vector, which should be chosen from the feasible set $\mathcal{A}$, i.e.,
$\bm{\alpha}\left(t\right)\in\mathcal{A}\triangleq \{\bm{\alpha}\in\mathbb{R}_{+}^{N}\big|\sum_{i\in\mathcal{N}}\alpha_{i}\leq 1\}$ \cite{ZWang1510}.

\subsubsection{MEC Server Scheduling}

The MEC server is equipped with an $M$-core CPU, where the CPU cores could be heterogeneous and the set of CPU cores is denoted as $\mathcal{M}\triangleq \{1,\cdots,M\}$. Denote the CPU-cycle frequency of the $m$th CPU core in the $t$th time slot as $f_{C,m}\left(t\right)$, which should be less than its maximum value $f_{C_{m},\max}$. Thus, the power consumption of the CPU cores at the MEC server can be expressed as
\begin{equation}
p_{\rm{ser}}\left(t\right)= \sum_{m\in\mathcal{M}}\kappa_{{\rm{ser}},m} f_{C,m}^{3}\left(t\right),
\label{pwrCPU}
\end{equation}
where $\kappa_{{\rm{ser}},m}$ is the effective switched capacitance of the $m$th CPU core at the MEC server.

The CPU cycles offered by the server CPU can be allocated for the computation tasks offloaded from different mobile devices, i.e., the MEC server scheduling decision should satisfy the following constraint:
\begin{equation}
\sum_{n\in\mathcal{N}}D_{s,n}\left(t\right)L_{n}\leq \sum_{m\in\mathcal{M}}f_{C,m}\left(t\right)\tau,t\in\mathcal{T},
\label{MECserverScheduling}
\end{equation}
which means the number of CPU cycles needed for completing $\mathbf{D}_{s}\left(t\right)\triangleq \left[D_{s,1}\left(t\right),\cdots,D_{s,N}\left(t\right)\right]$ should be no larger than the available CPU cycles at the MEC server.

\section{Problem Formulation}
In this section, we will first introduce the performance metrics, namely, the average weighted sum power consumption of the MEC system and the average sum queue length of the task buffers. An average weighted sum power consumption minimization problem with a task buffer stability constraint will then be formulated.

\subsection{Performance Metrics}
We focus on the power consumption of the task execution processes in the local and server CPUs, as well as the transmit power for computation offloading. Energy consumed for other purposes, e.g., powering the screens of the mobile devices and supporting the basic operations in the MEC system, is ignored for simplicity. Therefore, we adopt the average weighted sum power consumption of different entities in the MEC system as the performance metric, which is defined as follows:
\begin{equation}
\overline{P}_{\Sigma}\triangleq \lim_{T\rightarrow +\infty} \frac{1}{T}\sum_{t=0}^{T-1}\mathbb{E}\left[\sum_{i\in\mathcal{N}}w_{i}\left(p_{{\rm{tx}},i}\left(t\right)+p_{l,i}\left(t\right)\right)+w_{N+1}p_{\rm{ser}}\left(t\right)\right],
\label{netwpwr}
\end{equation}
where $w_{i}\geq 0,i\in\mathcal{N}$ is the weight of the power consumption at the $i$th mobile device and $w_{N+1}\geq 0$ is the weight of the power consumption at the MEC server. {These parameters can be adjusted to address the cost of power consumption at different nodes in the MEC system, as well as to balance the power consumption of the mobile devices and the MEC server \cite{Ge1208}, which are assumed to be constants throughout this paper.} For convenience, we denote
$P_{\Sigma}\left(t\right)\triangleq \sum_{i\in\mathcal{N}}w_{i}\left(p_{{\rm{tx}},i}\left(t\right)+p_{l,i}\left(t\right)\right)+w_{N+1}p_{\rm{ser}}\left(t\right)$,
i.e., $\overline{P}_{\Sigma}=\lim_{T\rightarrow +\infty}\frac{1}{T}\sum_{t=0}^{T-1}\mathbb{E}\left[P_{\Sigma}\left(t\right)\right]$.

According to \emph{Little's Law} \cite{Queuetheory}, the average execution delay experienced by each mobile device is proportional to the average number of its tasks waiting in the MEC system, which is the sum queue length of the task buffers at the device and server sides. Thus, the average sum queue length of the task buffers for each mobile device is used as a measurement of the execution delay, which can be written as
\begin{equation}
\overline{q}_{\Sigma,i}=\lim_{T\rightarrow +\infty}\frac{1}{T}\sum_{t=0}^{T-1}\mathbb{E}\left[Q_{i}\left(t\right)+T_{i}\left(t\right)\right],i\in\mathcal{N}.
\end{equation}

\subsection{Average Weighted Sum Power Consumption Minimization}

We denote the system operation at the $t$th time slot as $\mathbf{X}\left(t\right)\!\triangleq\!\left[\mathbf{f}\left(t\right)\!,\mathbf{p}_{\rm{tx}}\left(t\right)\!,\bm{\alpha}\left(t\right)\!,\mathbf{f}_{C}\left(t\right)\!,\mathbf{D}_{s}\left(t\right)\right]$, where $\mathbf{f}\left(t\right)\!\triangleq\!\left[f_{1}\left(t\right),\!\cdots\!,f_{N}\left(t\right)\right]$, $\mathbf{p}_{{\rm{tx}}}\left(t\right)\!\triangleq\!\left[p_{{\rm{tx}},1}\left(t\right),\!\cdots\!,p_{{\rm{tx}},N}\left(t\right)\right]$ and $\mathbf{f}_{C}\left(t\right)\!\triangleq\!\left[f_{C,1}\left(t\right),\!\cdots\!,f_{C,M}\left(t\right)\right]$. Therefore, the average weighted sum power consumption minimization problem can be formulated in $\mathbf{P_{1}}$, where (\ref{alphadomain}) is the bandwidth allocation constraint, while (\ref{freqtxconstraint}) denotes the CPU-cycle frequency and the transmit power constraints for the mobile devices. (\ref{Serverfreqconstraint}) is the CPU-cycle frequencies constraint for the $M$ CPU cores at the MEC server. The MEC server scheduling constraint is imposed by (\ref{MECserverSchedulingP1}). (\ref{stabilityconstraint}) enforces the task buffers to be mean rate stable \cite{Neely10}, which guarantees that all the arrived computation tasks can be completed with finite delay.
\begin{align}
&\mathbf{P_{1}}: \min_{\{\mathbf{X}\left(t\right)\}}\ \ \overline{P}_{\Sigma}\nonumber\\
&\ \ \ \ \ \ \ \mathrm{s.t.\ \ }\bm{\alpha}\left(t\right)\in\mathcal{A},t\in\mathcal{T}\label{alphadomain}\\
&\ \ \ \ \ \ \ \ \ \ \ \ \ 0\leq f_{i}\left(t\right)\leq f_{i,\max}, 0\leq p_{{\rm{tx}},i}\left(t\right)\leq p_{i,\max},i\in\mathcal{N},t\in\mathcal{T} \label{freqtxconstraint}\\
&\ \ \ \ \ \ \ \ \ \ \ \ \ 0\leq f_{C,m}\left(t\right)\leq f_{C_{m},\max},m\in\mathcal{M},t\in\mathcal{T} \label{Serverfreqconstraint}\\
&\ \ \ \ \ \ \ \ \ \ \ \ \ \sum_{n\in\mathcal{N}}D_{s,n}\left(t\right)L_{n}\leq \sum_{m\in\mathcal{M}}f_{C,m}\left(t\right)\tau,D_{s,i}\left(t\right)\geq 0, i\in\mathcal{N},t\in\mathcal{T}\label{MECserverSchedulingP1}\\
&\ \ \ \ \ \ \ \ \ \ \ \ \lim_{T\rightarrow +\infty}\frac{\mathbb{E}\left[|Q_{i}\left(T\right)|\right]}{T}=0,\lim_{T\rightarrow +\infty}\frac{\mathbb{E}\left[|T_{i}\left(T\right)|\right]}{T}=0,i\in\mathcal{N}. \label{stabilityconstraint}
\end{align}

\begin{rmk}
It is not difficult to identify that $\mathbf{P_{1}}$ is a stochastic optimization problem, for which, the operations at both the mobile device side (including the CPU-cycle frequencies for the local CPUs, as well as the transmit power and bandwidth allocation for computation offloading) and the MEC server side (including the MEC server scheduling and the CPU-cycle frequencies for the multiple CPU cores) need to be determined at each time slot. {This is a highly challenging problem with a large amount of SI (including the channel and task buffer state information) to be handled and a large number of variables to be determined. Also, the optimal decisions are temporally correlated due to the randomly arrived tasks.} Besides, an efficient resource management policy for the mobile devices and the MEC server is critical since offloading the computation tasks in either an over-conservative or an over-aggressive manner will result in ineffective use of the available computational resources. Moreover, the spatial coupling of the bandwidth allocation among different mobile devices poses an additional challenge for the radio resource management, which is also interdependent with the computational resource management. Therefore, a joint optimization of the radio and computational resource allocation is essential.
\end{rmk}

Instead of solving $\mathbf{P_{1}}$ directly, we consider its modified version, denoted as $\mathbf{P_{2}}$, which is obtained by replacing $\mathcal{A}$ in (\ref{alphadomain}) by $\tilde{\mathcal{A}}$ with
$\tilde{\mathcal{A}}\triangleq\{\bm{\alpha}\in\mathbb{R}_{+}^{N}|\sum_{i\in\mathcal{N}}\alpha_{i}\leq 1,\alpha_{i}\geq\epsilon_{A},i\in\mathcal{N}\}$, $\epsilon_{A}\in\left(0,1\slash N\right)$.
Thus, the task departure function of computation offloading, i.e., $D_{r,i}\left(t\right)$, is continuous and differentiable with respect to $\bm{\alpha}\left(t\right)\in\tilde{\mathcal{A}}$, which helps to develop an efficient asymptotically optimal online algorithm for $\mathbf{P_{2}}$ that is also feasible for $\mathbf{P_{1}}$. Besides, although the optimal value of $\mathbf{P_{2}}$ is larger than that of $\mathbf{P_{1}}$, they can be made arbitrarily close by setting $\epsilon_{A}$ to be sufficiently small. As a result, we will focus on $\mathbf{P_{2}}$ in the remainder of this paper.

\section{Online Joint Radio and Computational Resource Management Algorithm}
In this section, we will propose an online joint radio and computational resource management algorithm to solve $\mathbf{P_{2}}$ based on Lyapunov optimization \cite{Neely10}. {With the assistance of the Lyapunov optimization framework, we are able to resolve this challenging stochastic optimization problem by solving a deterministic per-time slot problem at each time slot, for which, the optimal solution can be obtained with low complexity.} Besides, a delay-improved mechanism will be designed for the proposed Lyapunov optimization-based algorithm. In the next section, we will show the proposed algorithm and its delay-improved version are capable of achieving asymptotic optimality and reveal the power-delay tradeoff in multi-user MEC systems.

\subsection{The Lyapunov Optimization-Based Online Algorithm}
To present the algorithm, we first define the Lyapunov function as
\begin{equation}
L\left(\bm{\Theta}\left(t\right)\right)=\frac{1}{2}\sum_{i\in\mathcal{N}}\left[Q_{i}^{2}\left(t\right)+T_{i}^{2}\left(t\right)\right],
\label{Lyvfunc}
\end{equation}
where $\bm{\Theta}\left(t\right)\triangleq \left[\mathbf{Q}\left(t\right),\mathbf{T}\left(t\right)\right]$.
Thus, the conditional Lyapunov drift can be written as
\begin{equation}
\Delta\left(\bm{\Theta}\left(t\right)\right)=
\mathbb{E}\left[L\left(\bm{\Theta}\left(t+1\right)\right)-L\left(\bm{\Theta}\left(t\right)\right)|\bm{\Theta}\left(t\right)\right].
\label{Lyvdrift}
\end{equation}
Accordingly, the Lyapunov drift-plus-penalty function can be expressed as
\begin{equation}
\Delta_{V}\left(\bm{\Theta}\left(t\right)\right)=\Delta\left(\bm{\Theta}\left(t\right)\right)+V\cdot\mathbb{E}\left[P_{\Sigma}\left(t\right)|\bm{\Theta}\left(t\right)\right],
\label{Lyvdriftpenalty}
\end{equation}
where $V\in\left(0,+\infty\right)$ (${\rm{bits}}^{2}\cdot {\rm{W}}^{-1}$) is a control parameter in the proposed algorithm. We first find an upper bound of $\Delta_{V}\left(\bm{\Theta}\left(t\right)\right)$ under any feasible $\mathbf{X}\left(t\right)$, as specified in Lemma \ref{Lyvdriftpenaltyboundlma}.

\begin{lma}
For arbitrary $\mathbf{X}\left(t\right)$ such that $f_{i}\left(t\right)\in\left[0,f_{i,\max}\right]$, $p_{{\rm{tx}},i}\left(t\right)\in\left[0,p_{i,\max}\right]$, $D_{s,i}\left(t\right)\geq 0, i\in\mathcal{N}$, $f_{C,m}\left(t\right)\in\left[0,f_{C_{m},\max}\right]$,  $m\in\mathcal{M}$, $\sum_{i\in\mathcal{N}}D_{s,i}\left(t\right)L_{i}\leq \sum_{m\in\mathcal{M}}f_{C,m}\left(t\right)\tau$, and $\bm{\alpha}\left(t\right)\in\tilde{\mathcal{A}}$, $\Delta_{V}\left(t\right)$ is upper bounded, i.e.,
\begin{equation}
\begin{split}
\Delta_{V}\left(\bm{\Theta}\left(t\right)\right)&\leq C - \mathbb{E}\left[\sum_{i\in\mathcal{N}}Q_{i}\left(t\right)\left(D_{\Sigma,i}\left(t\right)-A_{i}\left(t\right)\right)|\bm{\Theta}\left(t\right)\right]\\
&-\mathbb{E}\left[\sum_{i\in\mathcal{N}}T_{i}\left(t\right)\left(D_{s,i}\left(t\right)-D_{r,i}\left(t\right)\right)|\bm{\Theta}\left(t\right)\right]
+V\cdot \mathbb{E}\left[P_{\Sigma}\left(t\right)|\bm{\Theta}\left(t\right)\right],
\end{split}
\label{Lyvdriftpenaltybound}
\end{equation}
where $C$ is a constant.
\label{Lyvdriftpenaltyboundlma}
\end{lma}
\begin{proof}
Please refer to Appendix A.
\end{proof}

The main idea of the proposed online joint radio and computational resource management algorithm is to minimize the upper bound of $\Delta_{V}\left(\bm{\Theta}\left(t\right)\right)$ in the right-hand side of (\ref{Lyvdriftpenaltybound}) at each time slot. By doing so, the amount of tasks waiting in the task buffers can be maintained at a low level, meanwhile, the weighted sum power consumption of the mobile devices and the MEC server can be minimized. The proposed algorithm is summarized in Algorithm \ref{Algframework}, where a deterministic optimization problem $\mathbf{P}_{\rm{PTS}}$ needs to be solved at each time slot. It is worthy to note that the objective function of $\mathbf{P}_{\rm{PTS}}$ corresponds to the right-hand side of (\ref{Lyvdriftpenaltybound})\footnote{The terms that are not affected by $\mathbf{X}\left(t\right)$ is omitted in the objective function of $\mathbf{P}_{\rm{PTS}}$.}, and all the constraints in $\mathbf{P_{2}}$ except the task buffer stability constraint in (\ref{stabilityconstraint}) are retained in $\mathbf{P}_{\rm{PTS}}$. The optimal solution for $\mathbf{P}_{\rm{PTS}}$ will be developed in the next subsection.

\begin{algorithm}[t!]
\caption{The Online Joint Radio and Computational Resource Management Algorithm}
\label{alg1}
\begin{algorithmic}[1]
\STATE At the beginning of the $t$th time slot, obtain $\bm{\Theta}\left(t\right)$, $\{\Gamma_{i}\left(t\right)\}$, and $\{A_{i}\left(t\right)\}$.
\STATE Determine $\mathbf{f}\left(t\right), \mathbf{p}_{\rm{tx}}\left(t\right)$, $\bm{\alpha}\left(t\right)$, $\mathbf{D}_{s}\left(t\right)$, and $\textbf{f}_{C}\left(t\right)$ by solving
\begin{align}
&\mathbf{P}_{\rm{PTS}}:\min_{\mathbf{X}\left(t\right)} \ -\sum_{i\in\mathcal{N}}Q_{i}\left(t\right)D_{\Sigma,i}\left(t\right)-\sum_{i\in\mathcal{N}}T_{i}\left(t\right)\left(D_{s,i}\left(t\right)-D_{r,i}\left(t\right)\right)+V\cdot P_{\Sigma}\left(t\right)\nonumber\\
&\ \ \ \ \ \ \ \ \ \mathrm{s.t.}\ \ \ \bm{\alpha}\left(t\right)\in\tilde{\mathcal{A}}\ {\text{and}}\ (\ref{freqtxconstraint})-(\ref{MECserverSchedulingP1}),\nonumber
\nonumber
\end{align}
{where $P_{\Sigma}\left(t\right)\triangleq \sum_{i\in\mathcal{N}}w_{i}\left(p_{{\rm{tx}},i}\left(t\right)+p_{l,i}\left(t\right)\right)+w_{N+1}p_{\rm{ser}}\left(t\right)$.}
\STATE Update $\{Q_{i}\left(t\right)\}$ and $\{T_{i}\left(t\right)\}$ according to (\ref{bufferdynamics}) and (\ref{bufferdynamicsMEC}), respectively.
\STATE Set $t=t+1$.
\end{algorithmic}
\label{Algframework}
\end{algorithm}

\subsection{Optimal Solution For $\mathbf{P}_{\rm{PTS}}$}

In this subsection, we will develop the optimal solution for $\mathbf{P}_{\rm{PTS}}$, including the optimal CPU-cycle frequencies for the local CPUs, the transmit power and bandwidth allocation for computation offloading, as well as the MEC server scheduling and the CPU-cycle frequencies for the CPU cores at the MEC server. It can be seen that $\mathbf{P}_{\rm{PTS}}$ can be solved optimally by solving three sub-problems.

\textbf{Optimal CPU-Cycle Frequencies Of The Local CPUs:} It is straightforward to show that the optimal CPU-cycle frequencies for the local CPUs in time slot $t$ can be obtained by solving the following sub-problem $\mathbf{SP_{1}}$:
\begin{equation}
\begin{split}
&\mathbf{SP_{1}:}\min_{\mathbf{f}\left(t\right)}\ \sum_{i\in\mathcal{N}}\left(-Q_{i}\left(t\right) \tau f_{i}\left(t\right) L_{i}^{-1} + V \cdot w_{i}\kappa_{{\rm{mob}},i} f^{3}_{i}\left(t\right)\right)\\
&\ \ \ \ \ \ \ \ {\rm{s.t.}}\ \ 0\leq f_{i}\left(t\right)\leq f_{i,\max},i\in\mathcal{N}.
\end{split}
\end{equation}

First, since the objective function of $\mathbf{SP_{1}}$ is convex and its constraints are linear, $\mathbf{SP_{1}}$ is a convex optimization problem. Besides, as both the objective function and constraints of $\mathbf{SP_{1}}$ can be decomposed for individual $f_{i}\left(t\right)$, the optimization of $f_{i}\left(t\right)$ can be done separately at each mobile device. Therefore, the optimal $f^{\star}_{i}\left(t\right)$ is achieved at either the stationary point of $-Q_{i}\left(t\right) \tau f_{i}\left(t\right) L_{i}^{-1} + V \cdot w_{i}\kappa_{{\rm{mob}},i} f^{3}_{i}\left(t\right)$ or one of the boundary points, which is given by
\begin{equation}
f_{i}^{\star}\left(t\right)=
\begin{cases}
\min\bigg\{f_{i,\max},\sqrt{\frac{Q_{i}\left(t\right)\tau }{3\kappa_{{\rm{mob}},i}w_{i} V L_{i}}}\bigg\}, &w_{i}>0\\
f_{i,\max}, &w_{i}=0
\end{cases},i\in\mathcal{N}.
\end{equation}
\begin{rmk}
Note that $f_{i}^{\star}\left(t\right)$ is non-decreasing with $Q_{i}\left(t\right)$, as it is desirable to execute more tasks in order to keep the queue length of the local task buffer small. Besides, $f^{\star}_{i}\left(t\right)$ decreases with $V$, $L_{i}$, $\kappa_{{\rm{mob}},i}$ and $w_{i}$. In particular, with a larger value of $V$, the weight of $P_{\Sigma}\left(t\right)$ in $\mathbf{P}_{\rm{PTS}}$ becomes larger, and thus the local CPU slows down its frequency to reduce power consumption. With a larger value of $w_{i}$, the system places more emphasis on minimizing the power consumption at the $i$th mobile device, which also contributes to the reduction of $f_{i}^{\star}\left(t\right)$. On the other hand, with a larger value of $L_{i}$ ($\kappa_{{\rm{mob}},i}$), local execution becomes less effective as more CPU cycles (Joule of energy) will be needed to process per bit of task input, which again leads to a smaller CPU-cycle frequency. {It is worthwhile to mention that although $\mathbf{f}^{\star}\left(t\right)$ depends only on $\mathbf{Q}\left(t\right)$, the optimal local CPU speeds are temporally correlated with the optimal computation offloading decisions as will be derived in the sequel, since both of them affect the task buffer dynamics.}
\end{rmk}

\textbf{Optimal Transmit Power And Bandwidth Allocation:} The optimal $\mathbf{p}^{\star}_{\rm{tx}}\left(t\right)$ and $\bm{\alpha}^{\star}\left(t\right)$ can be obtained by solving the following sub-problem:
\begin{equation}
\begin{split}
&\mathbf{SP_{2}:}\min_{\bm{\alpha}\left(t\right),\mathbf{p}_{{\rm{tx}}}\left(t\right)}-\sum_{i\in\mathcal{N}}\left(Q_{i}\left(t\right)-T_{i}\left(t\right)\right)D_{r,i}\left(t\right)
+V\cdot \sum_{i\in\mathcal{N}}w_{i}p_{{\rm{tx}},i}\left(t\right)\\
&\ \ \ \ \ \ \ \ \ \ {\mathrm{s.t.}}\ \ \ \ 0\leq p_{{\rm{tx}},i}\left(t\right)\leq p_{i,\max},i\in\mathcal{N}\ \text{and}\ \bm{\alpha}\left(t\right)\in\tilde{\mathcal{A}},
\end{split}
\end{equation}
which is non-convex in general. In order to solve $\mathbf{SP_{2}}$, we identify an important property of its optimal solution, as shown in the following lemma.
\begin{lma}
For the set of mobile devices with $Q_{i}\left(t\right)\leq T_{i}\left(t\right)$ (denoted as $\tilde{\mathcal{N}}\left(t\right)\triangleq \{i|i\in\mathcal{N},Q_{i}\left(t\right)\leq T_{i}\left(t\right)\}$), the optimal transmit power and bandwidth allocation are given by $p_{{\rm{tx}},i}^{\star}\left(t\right)=0$ and $\alpha_{i}^{\star}\left(t\right)=\epsilon_{A}$, $i\in\tilde{\mathcal{N}}\left(t\right)$.
\label{propertySP2}
\end{lma}
\begin{proof}
For arbitrary $i\in\tilde{\mathcal{N}}\left(t\right)$, as $D_{r,i}\left(t\right)$ is a non-decreasing function of $p_{{\rm{tx}},i}\left(t\right)$, the term $-\left(Q_{i}\left(t\right)-T_{i}\left(t\right)\right)D_{r,i}\left(t\right)+V\cdot w_{i}p_{{\rm{tx}},i}\left(t\right)$ in the objective function of $\mathbf{SP_{2}}$ is non-decreasing with $p_{{\rm{tx}},i}\left(t\right)$, i.e., $p_{{\rm{tx}},i}^{\star}\left(t\right)=0$. Besides, with $p_{{\rm{tx}},i}^{\star}\left(t\right)=0$, $\alpha_{i}\left(t\right)$ has no contribution to the value of the objective function in $\mathbf{SP_{2}}$. Thus, the portion of bandwidth exceeding $\epsilon_{A}$ that is originally allocated for the $i$th mobile device can be re-allocated for the devices in $\tilde{\mathcal{N}}^{\rm{c}}\left(t\right)\triangleq\mathcal{N}\setminus \tilde{\mathcal{N}}\left(t\right)$, i.e., $\alpha^{\star}_{i}\left(t\right)=\epsilon_{A}, i\in\tilde{\mathcal{N}}\left(t\right)$.
\end{proof}

Lemma \ref{propertySP2} implies that a mobile device will offload only when the amount of tasks in the local task buffer is greater than that in its task buffer at the server. Intuitively, this is because when $T_{i}\left(t\right)\geq Q_{i}\left(t\right)$, the amount of tasks waiting in the MEC server is relatively large, and offloading tasks not only incurs transmit power consumption, but also brings detrimental effect on the execution delay, and thus it is inferior to local execution. Hence, by ruling out the mobile devices in $\tilde{\mathcal{N}}\left(t\right)$, we simplify $\mathbf{SP_{2}}$ as follows:
\begin{equation}
\begin{split}
&\mathbf{SP}'_{\mathbf{2}}\mathbf{:}\min_{\alpha_{i}\left(t\right),{p}_{{\rm{tx}},i}\left(t\right),i\in\tilde{\mathcal{N}}^{\rm{c}}\left(t\right)}-\sum_{i\in\tilde{\mathcal{N}}^{\rm{c}}\left(t\right)}\left(Q_{i}\left(t\right)-T_{i}\left(t\right)\right)D_{r,i}\left(t\right)
+V\cdot \sum_{i\in\tilde{\mathcal{N}}^{\rm{c}}\left(t\right)}w_{i}p_{{\rm{tx}},i}\left(t\right)\\
&\ \ \ \ \ \ \ \ \ \ \ \ \ \ {\mathrm{s.t.}}\ \ \ \ \ \ \ \ \ 0\leq p_{{\rm{tx}},i}\left(t\right)\leq p_{i,\max},i\in\tilde{\mathcal{N}}^{\rm{c}}\left(t\right)\\
&\ \ \ \ \ \ \ \ \ \ \ \ \ \ \ \ \ \ \ \ \ \ \ \ \ \ \  {\alpha}_{i}\left(t\right)\geq \epsilon_{A}, i\in\tilde{\mathcal{N}}^{\rm{c}}\left(t\right), \sum_{i\in\tilde{\mathcal{N}}^{\rm{c}}\left(t\right)} {\alpha}_{i}\left(t\right)\leq 1- |\tilde{\mathcal{N}}\left(t\right)|\cdot \epsilon_{A}.
\end{split}
\end{equation}

Since $D_{i}\left(t\right)\triangleq \tau \omega\log_{2}\left(1+\Gamma_{i}\left(t\right)p_{{\rm{tx}},i}\left(t\right)\left(N_{0}\omega\right)^{-1}\right)$ is concave with respect to $p_{{\rm{tx}},i}\left(t\right)$ and $D_{r,i}\left(t\right)$ is a perspective function of $D_{i}\left(t\right)$, i.e., $D_{r,i}\left(t\right)=\alpha_{i}\left(t\right)D_{i}\left(t\right)|_{p_{{\rm{tx}},i}\left(t\right)\slash \alpha_{i}\left(t\right)}$, $D_{r,i}\left(t\right)$ is jointly concave with respect to $\alpha_{i}\left(t\right)$ and $p_{{\rm{tx}},i}\left(t\right)$. Therefore,
$\mathbf{SP}'_{\mathbf{2}}$ is a convex optimization problem, and thus standard convex algorithms, such as the \emph{interior point method} \cite{Boyd04}, can be applied for the optimal solution. However, generic convex algorithms suffer from relatively high complexity, as they are developed for general convex problems and do not make full use of the problem structures. Motivated by this, we propose to solve $\mathbf{SP}'_{\mathbf{2}}$ by optimizing the transmit power and the bandwidth allocation in an alternating manner, where in each iteration, the optimal transmit powers are obtained in closed forms and the optimal bandwidth allocation is determined by the \emph{Lagrangian method}. Since $\mathbf{SP}'_{\mathbf{2}}$ is convex and its feasible region is a Cartesian product of those of $\{p_{{\rm{tx}},i}\left(t\right)\}$ and $\{\alpha_{i}\left(t\right)\}$, the alternating minimization procedure is guaranteed to converge to the global optimal solution, which is termed as the \emph{Gauss-Seidel method} in literature \cite{Grippo0004}. Next, we derive the solution in each iteration of the Gauss-Seidel method for fixed bandwidth and transmit power allocation, respectively.

\textbf{1) Optimal Transmit Power:} For a fixed bandwidth allocation, $\{\alpha_{i}\left(t\right)\},i\in\tilde{\mathcal{N}}^{\rm{c}}\left(t\right)$, the optimal transmit power for mobile devices in $\tilde{\mathcal{N}}^{\rm{c}}\left(t\right)$ can be obtained by solving
\begin{equation}
\begin{split}
&\mathbf{P}_{\rm{PWR}}:\min_{p_{{\rm{tx}},i}\left(t\right),i\in\tilde{\mathcal{N}}^{\rm{c}}\left(t\right)}-\sum_{i\in\tilde{\mathcal{N}}^{\rm{c}}\left(t\right)}\left(Q_{i}\left(t\right)-T_{i}\left(t\right)\right)D_{r,i}\left(t\right)+V\cdot \sum_{i\in\tilde{\mathcal{N}}^{\rm{c}}\left(t\right)}
w_{i}p_{{\rm{tx}},i}\left(t\right)\\
&\ \ \ \ \ \ \ \ \ \ \ \ \ \ {\rm{s.t.}}\ \ \ \ \ \ 0\leq p_{{\rm{tx}},i}\left(t\right)\leq p_{i,\max},i\in\tilde{\mathcal{N}}^{\rm{c}}\left(t\right).
\end{split}
\end{equation}

Similar to $\mathbf{SP}_{\mathbf{1}}$, $\mathbf{P}_{\rm{PWR}}$ can be decomposed for individual mobile device, and the optimal $p^{\star}_{{\rm{tx}},i}\left(t\right)$ is achieved at either the stationary point of $-\left(Q_{i}\left(t\right)-T_{i}\left(t\right)\right)D_{r,i}\left(t\right)+V\cdot w_{i}p_{{\rm{tx}},i}\left(t\right)$ or one of the boundary points, which is given in closed form by
\begin{equation}
p^{\star}_{{\rm{tx}},i}\left(t\right)=
\begin{cases}
\min\bigg\{\alpha_{i}\left(t\right)\omega\max\bigg\{\frac{\left(Q_{i}\left(t\right)-T_{i}\left(t\right)\right)\tau}{\ln2 \cdot V\cdot w_{i}}-\frac{N_{0}}{\Gamma_{i}\left(t\right)},0\bigg\},p_{i,\max}\bigg\}, &w_{i}>0\\
p_{i,\max}, &w_{i}=0
\end{cases},i\in\tilde{\mathcal{N}}^{\rm{c}}\left(t\right).
\label{optPWR}
\end{equation}

\textbf{2) Optimal Bandwidth Allocation:} For a fixed transmit power allocation $\{{p}_{{\rm{tx}},i}\left(t\right)\}$, $i\in\tilde{\mathcal{N}}^{\rm{c}}\left(t\right)$, the optimal bandwidth allocation can be obtained by solving the following problem:
\begin{equation}
\begin{split}
&\mathbf{P}_{\rm{BW}}:\min_{\alpha_{i}\left(t\right),i\in\tilde{\mathcal{N}}^{\rm{c}}\left(t\right)} \ -\sum_{i\in\tilde{\mathcal{N}}^{\rm{c}}\left(t\right)}\left(Q_{i}\left(t\right)-T_{i}\left(t\right)\right)D_{r,i}\left(t\right)\\
&\ \ \ \ \ \ \ \ \ \ \ \ {\rm{s.t.}}\ \ \ \ \ \ \alpha_{i}\left(t\right)\geq \epsilon_{A}, i\in\tilde{\mathcal{N}}^{\rm{c}}\left(t\right), \sum_{i\in\tilde{\mathcal{N}}^{\rm{c}}\left(t\right)}\alpha_{i}\left(t\right)\leq 1- |\tilde{\mathcal{N}}\left(t\right)|\cdot \epsilon_{A},
\end{split}
\end{equation}
which is more challenging as the bandwidth allocation is coupled among different mobile devices. Fortunately, the Lagrangian method offers an effective solution for $\mathbf{P}_{\rm{BW}}$. Specifically, the partial Lagrangian of $\mathbf{P}_{\rm{BW}}$ can be written as\footnote{We slightly abuse notation by using $\bm{\alpha}\left(t\right)$ to denote $\alpha_{i}\left(t\right),i\in\tilde{\mathcal{N}}^{\rm{c}}\left(t\right)$.}
\begin{equation}
\mathcal{L}\left(\bm{\alpha}\left(t\right),\lambda\left(t\right)\right)=-\sum_{i\in\tilde{\mathcal{N}}^{\rm{c}}\left(t\right)}\left(Q_{i}\left(t\right)-T_{i}\left(t\right)\right)D_{r,i}\left(t\right)+\lambda\left(t\right) \left[\sum_{i\in\tilde{\mathcal{N}}^{\rm{c}}\left(t\right)}\alpha_{i}\left(t\right)-\left(1- |\tilde{\mathcal{N}}\left(t\right)|\cdot \epsilon_{A}\right)\right],
\end{equation}
where $\lambda\left(t\right)\geq 0$ is the Lagrangian multiplier associated with $\sum_{i\in\tilde{\mathcal{N}}^{\rm{c}}\left(t\right)}\alpha_{i}\left(t\right)\leq 1- |\tilde{\mathcal{N}}\left(t\right)|\cdot \epsilon_{A}$. When $\exists i \in \tilde{\mathcal{N}}^{\rm{c}}\left(t\right)$ such that $p_{{\rm{tx}},i}\left(t\right)\Gamma_{i}\left(t\right)>0$, based on the Karush-Kuhn-Tucker (KKT) conditions, the optimal bandwidth allocation $\{\alpha_{i}^{\star}\left(t\right)\},i\in \tilde{\mathcal{N}}^{\rm{c}}\left(t\right)$ and the optimal Lagrangian multiplier $\lambda^{\star}\left(t\right)$ should satisfy the following equation set:
\begin{equation}
\begin{cases}
&\alpha_{i}^{\star}\left(t\right)=\max\{\epsilon_{A},\mathcal{R}_{i}\left(\lambda^{\star}\left(t\right)\right)\},i\in\tilde{\mathcal{N}}^{\rm{c}}\left(t\right),\lambda^{\star}\left(t\right)>0\\
&\sum_{i\in\tilde{\mathcal{N}}^{\rm{c}}\left(t\right)}\alpha_{i}^{\star}\left(t\right)=1- |\tilde{\mathcal{N}}\left(t\right)|\cdot \epsilon_{A}.
\end{cases}
\label{KKT}
\end{equation}
In (\ref{KKT}), if $p_{{\rm{tx}},i}\left(t\right)\Gamma_{i}\left(t\right)=0$, we define $\mathcal{R}_{i}\left(\lambda\left(t\right)\right)\triangleq \epsilon_{A}$; Otherwise, $\mathcal{R}_{i}\left(\lambda\left(t\right)\right)$ denotes the root of $\frac{\partial\mathcal{L}\left(\bm{\alpha}\left(t\right),\lambda\left(t\right)\right)}{\partial \alpha_{i}\left(t\right)}=-\left(Q_{i}\left(t\right)-T_{i}\left(t\right)\right)\frac{d D_{r,i}\left(t\right)}{d \alpha_{i}\left(t\right)}+\lambda\left(t\right)=0$ for $\lambda\left(t\right)>0$, which is positive and unique as $\frac{dD_{r,i}\left(t\right)}{d\alpha_{i}\left(t\right)}$ decreases with $\alpha_{i}\left(t\right)$, $\lim_{\alpha_{i}\left(t\right)\rightarrow 0^{+}}\frac{dD_{r,i}\left(t\right)}{d\alpha_{i}\left(t\right)}=+\infty$, and $\lim_{\alpha_{i}\left(t\right)\rightarrow +\infty}\frac{dD_{r,i}\left(t\right)}{d\alpha_{i}\left(t\right)}=0$. Thus, it suggests a bisection search over $\left[\lambda_{L}\left(t\right),\lambda_{U}\left(t\right)\right]$ for the optimal $\lambda^{\star}\left(t\right)$. Here, $\lambda_{L}\left(t\right)$ and $\lambda_{U}\left(t\right)$ can be chosen as
\begin{equation}
\begin{cases}
&\lambda_{L}\left(t\right)=\max_{i\in\tilde{\mathcal{N}}^{\rm{c}}\left(t\right)}\left(Q_{i}\left(t\right)-T_{i}\left(t\right)\right)\frac{dD_{r,i}\left(t\right)}{d\alpha_{i}\left(t\right)}|_{\alpha_{i}\left(t\right)=1- |\tilde{\mathcal{N}}\left(t\right)|\cdot \epsilon_{A}}\\
&\lambda_{U}\left(t\right)=\max_{i\in\tilde{\mathcal{N}}^{\rm{c}}\left(t\right)}\left(Q_{i}\left(t\right)-T_{i}\left(t\right)\right)\frac{dD_{r,i}\left(t\right)}{d\alpha_{i}\left(t\right)}|
_{\alpha_{i}\left(t\right)=\epsilon_{A}},
\end{cases}
\end{equation}
which satisfy $\sum_{i\in\tilde{\mathcal{N}}^{\rm{c}}\left(t\right)}\max\{\epsilon_{A},\mathcal{R}_{i}\left(\lambda_{L}\left(t\right)\right)\}>1- |\tilde{\mathcal{N}}\left(t\right)|\cdot \epsilon_{A}$ and $\sum_{i\in\tilde{\mathcal{N}}^{\rm{c}}\left(t\right)}\max\{\epsilon_{A},\mathcal{R}_{i}\left(\lambda_{U}\left(t\right)\right)\}<1- |\tilde{\mathcal{N}}\left(t\right)|\cdot \epsilon_{A}$, respectively. Hence, $\mathcal{R}_{i}\left(\lambda\left(t\right)\right)$ can be obtained by a bisection search over $\left(0,1- |\tilde{\mathcal{N}}\left(t\right)|\cdot \epsilon_{A}\right]$ when $p_{{\rm{tx}},i}\left(t\right)\Gamma_{i}\left(t\right)>0$, and the searching process for the optimal $\lambda^{\star}\left(t\right)$ will be terminated when $|\sum_{i\in\tilde{\mathcal{N}}^{\rm{c}}\left(t\right)}\max\{\epsilon_{A},\mathcal{R}_{i}\left(\lambda\left(t\right)\right)\}-(1- |\tilde{\mathcal{N}}\left(t\right)|\cdot \epsilon_{A})|<\xi$, where $\xi$ is the accuracy of the algorithm. When $p_{{\rm{tx}},i}\left(t\right)\Gamma_{i}\left(t\right)=0,\forall i\in \tilde{\mathcal{N}}^{\rm{c}}\left(t\right)$, $\alpha_{i}^{\star}\left(t\right)=\epsilon_{A},i\in \tilde{\mathcal{N}}^{\rm{c}}\left(t\right)$ is the optimal bandwidth allocation. Details of the Lagrangian method for $\mathbf{P}_{\rm{BW}}$ are summarized in Algorithm \ref{Lagrangianmtd}.

{
\begin{rmk}
The proposed Gauss-Seidel method updates the transmit power and bandwidth allocation alternately in each iteration, which will converge to the optimal solution of $\mathbf{SP'_{2}}$ with a sublinear convergence rate \cite{ABeck15}. In each iteration, the major complexity comes from the Lagrangian method for the optimal bandwidth allocation, which employs bisection search for $\lambda^{\star}\left(t\right)$, and will terminate within $\log_{2}\left(\frac{\lambda_{U}\left(t\right)-\lambda_{L}\left(t\right)}{\lambda_{\xi}\left(t\right)}\right)$ iterations in time slot $t$ ($\lambda_{\xi}\left(t\right)$ corresponds to the accuracy requirement for $\lambda^{\star}\left(t\right)$ given $\xi$). Besides, in order to search for the optimal $\lambda^{\star}\left(t\right)$, $\log_{2}\left(\frac{1}{\varsigma}\right)$ evaluations for $\frac{\partial \mathcal{L}\left(\bm{\alpha}\left(t\right),\lambda\left(t\right)\right)}{\partial \alpha_{i}\left(t\right)}$ are needed in order to determine $\mathcal{R}_{i}\left(\lambda\right)$ for each mobile device, where $\varsigma$ is the accuracy requirement. Therefore, Algorithm \ref{Lagrangianmtd} will terminate within a finite number of evaluations for functions $\{\frac{\partial \mathcal{L}\left(\bm{\alpha}\left(t\right),\lambda\left(t\right)\right)}{\partial \alpha_{i}\left(t\right)}\}_{i\in\tilde{\mathcal{N}}^{\rm{c}}\left(t\right)}$, which is given by $|\tilde{\mathcal{N}}^{\rm{c}}\left(t\right)|\log_{2}\left(\frac{1}{\varsigma}\right)\log_{2}\left(\frac{\lambda_{U}\left(t\right)-\lambda_{L}\left(t\right)}{\lambda_{\xi}\left(t\right)}\right)$.
\end{rmk}
}

\begin{algorithm}[h]
\caption{Lagrangian Method for $\mathbf{P}_{\rm{BW}}$}
\begin{algorithmic}[1]
\STATE Set $\xi\!=\!10^{-7}$, $\tilde{\lambda}_{L}\!=\!\lambda_{L}\left(t\right)$, $\tilde{\lambda}_{U}=\lambda_{U}\left(t\right)$, $l=0$, $I_{\max}=200$, $\epsilon_{A}=10^{-4}$, $\alpha_{i}\left(t\right)=\epsilon_{A},i\in \tilde{\mathcal{N}}^{\rm{c}}\left(t\right)$.
\STATE \textbf{While} {$|\sum_{i\in\tilde{\mathcal{N}}^{\rm{c}}\left(t\right)}\alpha_{i}\left(t\right)-\left(1-|\tilde{\mathcal{N}}\left(t\right)|\cdot \epsilon_{A}\right)|\geq   \xi$} and $l\leq I_{\max}$ \textbf{do}
\STATE \hspace{10pt} $\tilde{\lambda}=\frac{1}{2}\left(\tilde{\lambda}_{L}+\tilde{\lambda}_{U}\right)$ and $l=l+1$.
\STATE \hspace{10pt} Set $\alpha_{i}\left(t\right)= \max\{\epsilon_{A},\mathcal{R}_{i}\left(\tilde{\lambda}\right)\}$, $i\in\tilde{\mathcal{N}}^{\rm{c}}\left(t\right)$.
\STATE \hspace{10pt} \textbf{If} {$\sum_{i\in\tilde{\mathcal{N}}^{\rm{c}}\left(t\right)}\alpha_{i}\left(t\right)> 1-|\tilde{\mathcal{N}}\left(t\right)|\cdot \epsilon_{A}$} \textbf{then}
\STATE \hspace{20pt} $\tilde{\lambda}_{L}=\tilde{\lambda}$.
\STATE \hspace{10pt} \textbf{Else}
\STATE \hspace{20pt} $\tilde{\lambda}_{U}=\tilde{\lambda}$.
\STATE \hspace{10pt} \textbf{Endif}
\STATE \textbf{Endwhile}
\end{algorithmic}
\label{Lagrangianmtd}
\end{algorithm}

\textbf{3) Optimal CPU-Cycle Frequencies And Scheduling At The MEC Server:} After decoupling $\mathbf{f}\left(t\right)$, $\mathbf{p}_{\rm{tx}}\left(t\right)$ and $\bm{\alpha}\left(t\right)$ from $\mathbf{P}_{\rm{PTS}}$, we find that the optimal CPU-cycle frequencies for the $M$ CPU cores at the MEC server $\mathbf{f}^{\star}_{C}\left(t\right)$ and the optimal scheduling decision $\mathbf{D}^{\star}_{s}\left(t\right)$ can be obtained by solving the following sub-problem:
\begin{equation}
\begin{split}
&\mathbf{SP_{3}:}\min_{\mathbf{f}_{C}\left(t\right),\mathbf{D}_{s}\left(t\right)}\ -\sum_{i\in\mathcal{N}}T_{i}\left(t\right)D_{s,i}\left(t\right)+V \cdot w_{N+1}\sum_{m\in\mathcal{M}}\kappa_{{\rm{ser}},m}f^{3}_{C,m}\left(t\right)\\
&\ \ \ \ \ \ \ \ \ \ {\rm{s.t.}}\ \ \ \ \ 0\leq f_{C,m}\left(t\right)\leq f_{C_{m},\max},m\in\mathcal{M}\\
&\ \ \ \ \ \ \ \ \ \ \ \ \ \ \ \ \ \ \sum_{n\in\mathcal{N}}D_{s,n}\left(t\right)L_{n}\leq \sum_{m\in\mathcal{M}}f_{C,m}\left(t\right)\tau,D_{s,i}\left(t\right)\geq 0,i\in\mathcal{N}.
\end{split}
\end{equation}

It is not difficult to verify that $\mathbf{SP_{3}}$ is a convex problem, and thus standard convex algorithms can be used. Interestingly, by identifying an important structure of the optimal solution as shown in the following lemma, we are able to solve $\mathbf{SP_{3}}$ in closed form.

\begin{lma}
For any feasible $\mathbf{f}_{C}\left(t\right)$, there exists an optimal solution for $\mathbf{SP_{3}}$ where at most one mobile device is scheduled, i.e., $|\mathbf{D}_{s}\left(t\right)|_{0}\leq 1$, and the mobile device being scheduled is the one with the highest value of $T_{i}\left(t\right)L_{i}^{-1}$,\footnote{$|\mathbf{x}|_{0}$ denotes the number of non-zero elements in the vector $\mathbf{x}$. When there are multiple devices with the same maximum value of $T_{i}\left(t\right)L_{i}^{-1}$, we schedule the one with the smallest device index, which preserves optimality. Therefore, we assume the mobile device with the highest value of $T_{i}\left(t\right)L_{i}^{-1}$ is unique in the following for ease of presentation.} whose device index is denoted as $i_{\mathcal{N}}^{\max}$.
\label{propertySP3}
\end{lma}
\begin{proof}
Please refer to Appendix B.
\end{proof}

Based on Lemma \ref{propertySP3}, we have $D^{\star}_{s,i}\left(t\right)=0,i\neq i_{\mathcal{N}}^{\max}$. Therefore, we can simplify $\mathbf{SP_{3}}$ as
\begin{equation}
\begin{split}
&\mathbf{SP'_{3}:}\min_{\mathbf{f}_{C}\left(t\right),D_{s,i_{\mathcal{N}}^{\max}}\left(t\right)}\ -T_{i_{\mathcal{N}}^{\max}}\left(t\right)D_{s,i_{\mathcal{N}}^{\max}}\left(t\right)+V\cdot w_{N+1}\sum_{m\in\mathcal{M}}\kappa_{{\rm{ser}},m}f^{3}_{C,m}\left(t\right)\\
&\ \ \ \ \ \ \ \ \ \ \ \ {\rm{s.t.}}\ \ \ \ \ 0\leq f_{C,m}\left(t\right)\leq f_{C_{m},\max},m\in\mathcal{M}\\
&\ \ \ \ \ \ \ \ \ \ \ \ \ \ \ \ \ \ \ \ D_{s,i_{\mathcal{N}}^{\max}}\left(t\right)L_{i_{\mathcal{N}}^{\max}}\leq \sum_{m\in\mathcal{M}}f_{C,m}\left(t\right)\tau,D_{s,i_{\mathcal{N}}^{\max}}\left(t\right)\geq 0,
\end{split}
\end{equation}
and its optimal solution is given by the following corollary.
\begin{corol}
The optimal solution for $\mathbf{SP'_{3}}$ (also for $\mathbf{SP_{3}}$) is given as
\begin{equation}
f^{\star}_{C,m}\left(t\right)=
\begin{cases}
\min\bigg\{f_{C_{m},\max},\sqrt{\frac{T_{i_{\mathcal{N}}^{\max}}\left(t\right)\tau}{3Vw_{N+1}L_{i^{\max}_{\mathcal{N}}}\kappa_{{\rm{ser}},m}}}\bigg\}, &w_{N+1}>0\\
f_{C_{m},\max}, &w_{N+1}=0
\end{cases},m\in\mathcal{M}.
\label{mcorefreq}
\end{equation}
and $D^{\star}_{s,i_{\mathcal{N}}^{\max}}\left(t\right)=L^{-1}_{i^{\max}_{\mathcal{N}}}\sum_{m\in\mathcal{M}}f^{\star}_{C,m}\left(t\right)\tau$ and $D_{s,i}^{\star}\left(t\right)=0,i\neq i_{\mathcal{N}}^{\max}$.
\label{solutionSP3pi}
\end{corol}

\begin{proof}
Since $T_{i_{\mathcal{N}}^{\max}}\left(t\right) \geq 0$, there exists an optimal solution for $\mathbf{SP'_{3}}$ such that $D_{s,i_{\mathcal{N}}^{\max}}\left(t\right)L_{i_{\mathcal{N}}^{\max}}= \sum_{m\in\mathcal{M}}f_{C,m}\left(t\right)\tau$. Thus, by substituting $D_{s,i_{\mathcal{N}}^{\max}}\left(t\right)=L^{-1}_{i^{\max}_{
\mathcal{N}}}\sum_{m\in\mathcal{M}}f_{C,m}\left(t\right)\tau$ into the objective function, we find that $f^{\star}_{C,m}\left(t\right)$ is the optimal solution for
\begin{equation}
\min_{0\leq f_{C,m}\left(t\right)\leq f_{C_{m},\max}}-T_{i_{\mathcal{N}}^{\max}}\left(t\right)L^{-1}_{i^{\max}_{\mathcal{N}}}f_{C,m}\left(t\right)\tau + V\cdot w_{N+1}\kappa_{{\rm{ser}},m}f_{C,m}^{3}
\left(t\right),
\end{equation}
which is given by (\ref{mcorefreq}). Together with Lemma \ref{propertySP3}, we complete the proof.
\end{proof}

\begin{rmk}
One main benefit of the proposed online algorithm is that it does not require prior information on the computation task arrival and wireless channel fading processes, which makes it also applicable for unpredictable environments. Besides, the proposed algorithm is of low complexity, as at each time slot, the optimal CPU-cycle frequencies of the local and server CPUs, as well as the MEC server scheduling decision, are obtained in closed forms, while the computation offloading policy is determined by an efficient alternating minimization algorithm. Furthermore, as will be shown in the next section, the achievable performance of the proposed algorithm can be analytically characterized, which helps to demonstrate its asymptotic optimality.
\label{benefitsalgorithm}
\end{rmk}

\subsection{A Delay-Improved Mechanism}

In Algorithm \ref{Algframework}, the system operation is determined by the optimal solution of the per-time slot problem $\mathbf{P}_{\rm{PTS}}$, which excludes the task buffer underflow constraint, i.e., some of the departure functions, either $D_{\Sigma,i}\left(t\right)$ or $D_{s,i}\left(t\right)$, may be greater than the actual amount of tasks in the corresponding buffers. In particular, as there is only one mobile device being scheduled by the MEC server at each time slot according to Lemma \ref{propertySP3} and Corollary \ref{solutionSP3pi}, part of the available CPU cycles at the MEC server may be wasted and the excessive computational resource can be re-allocated for other devices. Inspired by this observation, we propose a delay-improved mechanism for Algorithm \ref{Algframework} in this subsection.

In the proposed delay-improved mechanism, the decision center maintains a set of virtual task buffers $\bm{\Theta}_{vir}\left(t\right)$ with $\bm{\Theta}_{vir}\left(0\right)=\mathbf{0}$, and tracks the states of the actual task buffers $\bm{\Theta}_{act}\left(t\right)$. The knowledge of $\bm{\Theta}_{vir}\left(t\right)$ is used to compute the optimal solution of the per-time slot problem $\mathbf{X}^{\star}\left(t\right)$ according to Algorithm \ref{Algframework}. We modify the MEC server scheduling decision $\mathbf{D}^{\star}_{s}\left(t\right)$ in $\mathbf{X}^{\star}\left(t\right)$ as $\tilde{\mathbf{D}}_{s}\left(t\right)$ in the actual implementation according to the procedures shown in Algorithm \ref{Delayenhanced}.

\begin{algorithm}[th!]
\caption{The Delay-Improved Mechanism}
\begin{algorithmic}[1]
\STATE If $\sum_{m\in\mathcal{M}}f^{\star}_{C,m}\left(t\right)\tau\leq T_{i_{\mathcal{N}}^{\max},act}\left(t\right)L_{i_{\mathcal{N}}^{\max}}$, set $\tilde{\mathbf{D}}_{s}\left(t\right)=\mathbf{D}^{\star}_{s}\left(t\right)$ and exit, where $i_{\mathcal{N}}^{\max}=\arg\max_{i\in\mathcal{N}}T_{i,vir}\left(t\right)L_{i}^{-1}$; Otherwise, go to \textbf{Line 2}.
\STATE Sort the mobile devices according to the descending order of $T_{i,vir}\left(t\right)L_{i}^{-1}$ and let $\left[i\right]$ denote the mobile device with the $i$th highest value of $T_{i,vir}\left(t\right)L_{i}^{-1}$, i.e., $\left[1\right]=i_{\mathcal{N}}^{\max}$. (If there are multiple devices with the same values of $T_{i,vir}\left(t\right)L_{i}^{-1}$, sort them according to the ascending order of their device indices.)
\STATE If $\sum_{i\in\mathcal{N}}T_{i,act}\left(t\right)L_{i} > \sum_{m\in\mathcal{M}}f^{\star}_{C,m}\left(t\right)\tau$, determine the number of scheduled mobile devices as
$n_{\rm{ser}}\left(t\right)=
\min\big\{n|\sum_{i=1}^{n}T_{\left[i\right],act}\left(t\right)L_{\left[i\right]}>\sum_{m\in\mathcal{M}}f^{\star}_{C,m}\left(t\right)\tau\big\}$; Otherwise, $n_{\rm{ser}}\left(t\right)=N$.
\STATE Re-allocate the CPU cycles offered by the MEC server according to
\begin{equation}
\tilde{D}_{s,\left[n\right]}\left(t\right)=
\begin{cases}
T_{\left[n\right],act}\left(t\right), &n< n_{\rm{ser}}\left(t\right)\\
\left(\sum_{m\in\mathcal{M}}f^{\star}_{C,m}\left(t\right)\tau - \sum_{i<n_{\rm{ser}}\left(t\right)}T_{\left[i\right],act}\left(t\right)L_{\left[i\right]}\right)L^{-1}_{\left[n\right]}, & n= n_{\rm{ser}}\left(t\right)\\
0, &{\rm{otherwise}}.
\end{cases}
\nonumber
\end{equation}
\end{algorithmic}
\label{Delayenhanced}
\end{algorithm}

In the following proposition, we show the performance of the delay-improved mechanism, including the average weighted sum power consumption and the average execution delay, is no worse than that achieved by Algorithm \ref{Algframework}.

\begin{prop}
The power consumption of the mobile devices and the MEC server under the delay-improved mechanism, i.e., Algorithm \ref{Delayenhanced}, is the same as that under Algorithm \ref{Algframework}, while the average execution delay, i.e., the average sum queue length, is not increased.
\label{propmodialgo}
\end{prop}
\begin{proof}
The property of the power consumption performance is straightforward, while the proof for the delay performance can be obtained by verifying $Q_{i,act}\left(t\right)=Q_{i,vir}\left(t\right)$ and $T_{i,act}\left(t\right)\leq T_{i,vir}\left(t\right),i\in\mathcal{N}, t\in\mathcal{T}$ via mathematical induction, which is omitted for brevity.
\end{proof}

It is worthwhile to emphasize that the proposed delay-improved mechanism is able to secure the benefits of Algorithm \ref{Algframework} as mentioned in Remark \ref{benefitsalgorithm}, while enhancing the average execution delay performance without incurring extra power consumption.

\section{Performance Analysis}
In this section, we will provide the main theoretical results of this paper, which characterize the upper bounds for the average weighted sum power consumption of the mobile devices and the MEC server, as well as the average sum queue length of the task buffers at the mobile and server sides. Besides, the tradeoff between the average weighted sum power consumption and average execution delay will also be revealed.

To facilitate the analysis, we first define an auxiliary problem $\mathbf{P_{3}}$ with the same objective function and constraints as $\mathbf{P_{2}}$, but the task buffer dynamics at the MEC server is modified as
\begin{equation}
T_{i}\left(t+1\right)=\max\{T_{i}\left(t\right)-D_{s,i}\left(t\right),0\}+D_{r,i}\left(t\right),i\in\mathcal{N}.
\end{equation}

In the following lemma, we demonstrate the relationship between $\mathbf{P_{2}}$ and $\mathbf{P_{3}}$ in terms of their feasibility and optimal values of the objective functions.

\begin{lma}
$\mathbf{P_{2}}$ is feasible if and only if $\mathbf{P_{3}}$ is feasible, and the optimal values of $\mathbf{P_{2}}$ and $\mathbf{P_{3}}$, denoted as $P_{\Sigma,\mathbf{P_{2}}}^{\rm{opt}}$ and $P_{\Sigma,\mathbf{P_{3}}}^{\rm{opt}}$, respectively, are equal, i.e., $P_{\Sigma,\mathbf{P_{2}}}^{\rm{opt}}=P_{\Sigma,\mathbf{P_{3}}}^{\rm{opt}}$.
\label{equivalenceP2P3}
\end{lma}

\begin{proof}
We first show that if $\mathbf{P_{3}}$ is feasible, then $\mathbf{P_{2}}$ is also feasible.\footnote{As the instantaneous constraints, i.e., (\ref{freqtxconstraint})-(\ref{MECserverSchedulingP1}) and $\bm{\alpha}\left(t\right)\in \tilde{\mathcal{A}},t\in\mathcal{T}$, are the same for both problems, it is sufficient to verify the mean rate stability of the task buffers in order to show the feasibility of $\mathbf{P_{2}}$.} Denote the system operation under the optimal solution for $\mathbf{P_{3}}$ as $\{\mathbf{X}^{3\star}\left(t\right)\}$, and the queue lengths induced by this solution in $\mathbf{P_{3}}$ as $\bm{\Theta}^{{3}\star}_{\mathbf{P_{3}}}\left(t\right)=\left[\mathbf{Q}^{{3}\star}_{\mathbf{P_{3}}}\left(t\right),\mathbf{T}^{{3}\star}_{\mathbf{P_{3}}}\left(t\right)\right]$. By applying $\{\mathbf{X}^{3\star}\left(t\right)\}$ to $\mathbf{P_{2}}$, we can easily show that the queue lengths induced by $\{\mathbf{X}^{3\star}\left(t\right)\}$ in $\mathbf{P_{2}}$, $\bm{\Theta}^{3\star}_{\mathbf{P_{2}}}\left(t\right)=\left[\mathbf{Q}^{{3}\star}_{\mathbf{P_{2}}}\left(t\right),\mathbf{T}^{{3}\star}_{\mathbf{P_{2}}}\left(t\right)\right]$, satisfies $Q_{\mathbf{P_{2}},i}^{3\star}\left(t\right)=Q_{\mathbf{P_{3}},i}^{3\star}\left(t\right)$ and $T_{\mathbf{P_{2}},i}^{3\star}\left(t\right)\leq T_{\mathbf{P_{3}},i}^{3\star}\left(t\right),i\in\mathcal{N},t\in\mathcal{T}$, i.e., $\mathbf{P_{2}}$ is also feasible. Besides, the average weighted sum power consumption under $\{\mathbf{X}^{{3}\star}\left(t\right)\}$ in $\mathbf{P_{2}}$ equals $P_{\Sigma,\mathbf{P_{3}}}^{\rm{opt}}$, i.e., $P_{\Sigma,\mathbf{P_{2}}}^{\rm{opt}} \leq P_{\Sigma,\mathbf{P_{3}}}^{\rm{opt}}$.

Next, suppose $\mathbf{P_{2}}$ is feasible and its optimal solution is given by $\{\mathbf{X}^{2\star}\left(t\right)\}$, which can be utilized to construct a solution $\{\tilde{\mathbf{X}}\left(t\right)\}$ for $\mathbf{P_{3}}$ by modifying $\{D^{2\star}_{r,i}\left(t\right)\}$ as
\begin{equation}
\tilde{D}_{r,i}\left(t\right)=\min\{D^{2\star}_{r,i}\left(t\right),\max\{\tilde{Q}_{\mathbf{P_{3}},i}\left(t\right)-D^{2\star}_{l,i}\left(t\right),0\}\}\leq D^{2\star}_{r,i}\left(t\right),
\end{equation}
where $\tilde{\bm{\Theta}}_{\mathbf{P_{3}}}\left(t\right)=\left[\tilde{\mathbf{Q}}_{\mathbf{P_{3}}}\left(t\right), \tilde{\mathbf{T}}_{\mathbf{P_{3}}}\left(t\right)\right]$ denotes the queue lengths of the task buffers in $\mathbf{P_{3}}$ induced by $\{\tilde{\mathbf{X}}\left(t\right)\}$. By mathematical induction, we can show that $\tilde{\mathbf{Q}}_{\mathbf{P_{3}}}\left(t\right)=\mathbf{Q}^{2\star}_{\mathbf{P_{2}}}\left(t\right)$ and $\tilde{\mathbf{T}}_{\mathbf{P_{3}}}\left(t\right)=\mathbf{T}^{2\star}_{\mathbf{P_{2}}}\left(t\right)$, where $\bm{\Theta}^{2\star}_{\mathbf{P_{2}}}\left(t\right)=\left[\mathbf{Q}^{2\star}_{\mathbf{P_{2}}}\left(t\right), \mathbf{T}^{2\star}_{\mathbf{P_{2}}}\left(t\right)\right]$ denotes the queue lengths in $\mathbf{P_{2}}$ under $\{\mathbf{X}^{2\star}\left(t\right)\}$. Thus, $\mathbf{P_{3}}$ is also feasible as the transmit powers at the mobile devices for realizing $\tilde{D}_{r,i}\left(t\right)$ are decreased since $\tilde{D}_{r,i}\left(t\right)\leq D^{2\star}_{r,i}\left(t\right)$, i.e., $P_{\Sigma,\mathbf{P_{3}}}^{\rm{opt}}\leq\overline{P}^{\tilde{\mathbf{X}}}_{\Sigma,\mathbf{P_{3}}}\leq P_{\Sigma,\mathbf{P_{2}}}^{\rm{opt}}$, where $\overline{P}^{\tilde{\mathbf{X}}}_{\Sigma,\mathbf{P_{3}}}$ is the value of the objective function in $\mathbf{P_{3}}$ under $\tilde{\mathbf{X}}\left(t\right)$. As a result, we have $P_{\Sigma,\mathbf{P_{2}}}^{\rm{opt}}=P_{\Sigma,\mathbf{P_{3}}}^{\rm{opt}}$.
\end{proof}

In the following lemma, we show there exists a stationary and randomized policy, in which the system operations are i.i.d. among different time slots, that behaves arbitrarily close to the optimal solution for $\mathbf{P_{3}}$.
\begin{lma}
If $\mathbf{P_{3}}$ is feasible, then for any $\delta>0$, there exists a stationary and randomized policy $\Pi$ that satisfies all the instantaneous constraints, and
\begin{equation}
\begin{cases}
\mathbb{E}\left[P_{\Sigma}^{\Pi}\left(t\right)\right]\leq P_{\Sigma,\mathbf{P_{3}}}^{\rm{opt}}+\delta\\
\lambda_{i}\leq \mathbb{E}\left[D_{\Sigma,i}^{\Pi}\left(t\right)\right]+\delta,i\in\mathcal{N}\\
\mathbb{E}\left[D^{\Pi}_{r,i}\left(t\right)\right] \leq \mathbb{E}\left[D^{\Pi}_{s,i}\left(t\right)\right]+\delta,i\in\mathcal{N}.
\end{cases}
\end{equation}
\label{stationaryrandomizedproperty}
\end{lma}
\begin{proof}
The proof can be obtained by Theorem 4.5 in \cite{Neely10}, which is omitted for brevity.
\end{proof}

Based on Lemma \ref{equivalenceP2P3} and Lemma \ref{stationaryrandomizedproperty}, the worst-case average weighted sum power consumption and average sum queue length under Algorithm \ref{Algframework} and Algorithm \ref{Delayenhanced} are derived in the following theorem.
\begin{thm}
Assume that $\mathbf{P_{2}}$ is feasible, then under Algorithm \ref{Algframework} (also Algorithm \ref{Delayenhanced}), we have:
\begin{enumerate}
\item The average weighted sum power consumption satisfies:
\begin{equation}
\overline{P}_{\Sigma}^{\star}\leq P^{\rm{opt}}_{\Sigma,\mathbf{P_{2}}}+C\cdot V^{-1},
\end{equation}
where $\overline{P}^{\star}_{\Sigma}$ is the average weighted sum power consumption under the proposed algorithm.
\item For arbitrary $i\in\mathcal{N}$, $Q_{i}\left(t\right)$ and $T_{i}\left(t\right)$ are mean rate stable.
\item Suppose there exist $\epsilon>0$ and $\Psi\left(\epsilon\right)$ ($P_{\Sigma,\mathbf{P_{2}}}^{\rm{opt}}\leq \Psi\left(\epsilon\right)\leq \sum_{i\in\mathcal{N}}w_{i}\left(\kappa_{{\rm{mob}},i}f_{i,\max}^{3}+p_{i,\max}\right)+w_{N+1}\sum_{m\in\mathcal{M}}\kappa_{{\rm{ser}},m}f_{C_{m},\max}^{3}$) , and a stationary and randomized algorithm $\tilde{\Pi}$ that satisfies $\bm{\alpha}\left(t\right)\in \tilde{\mathcal{A}},t\in\mathcal{T}$, (\ref{freqtxconstraint})-(\ref{MECserverSchedulingP1}), and the following Slater conditions \cite{Neely10}:
    \begin{equation}
    \begin{cases}
    \mathbb{E}\left[P^{\tilde{\Pi}}_{\Sigma}\left(t\right)\right]=\Psi\left(\epsilon\right)\\
    \lambda_{i}\leq \mathbb{E}\left[D^{\tilde{\Pi}}_{\Sigma,i}\left(t\right)\right]-\epsilon,i\in\mathcal{N}\\
    \mathbb{E}\left[D_{r,i}^{\tilde{\Pi}}\left(t\right)\right]\leq \mathbb{E}\left[D^{\tilde{\Pi}}_{s,i}\left(t\right)\right]-\epsilon,i\in\mathcal{N},
    \end{cases}
    \end{equation}

     then the average sum queue length of the task buffers of all the mobile devices satisfies:
  \begin{equation}
    \sum_{i\in\mathcal{N}}\overline{q}_{\Sigma,i}\leq \frac{C+V\left(\Psi\left(\epsilon\right)-P_{\Sigma,\mathbf{P_{2}}}^{\rm{opt}}\right)}{\epsilon}.
   \end{equation}
\end{enumerate}
\label{performanalysis}
\end{thm}
\begin{proof}
Please refer to Appendix C.
\end{proof}

\begin{rmk}
Theorem \ref{performanalysis} shows that under the proposed online joint radio and computational resource management algorithm, the worst-case average weighted sum power consumption of the MEC system decreases inversely proportional to $V$, which demonstrates the asymptotic optimality of the proposed algorithm. It is also shown that the upper bound of the average sum queue length, i.e., the execution delay according to Little's Law, increases linearly with $V$, i.e., there exists an $\left[O\left(1\slash V\right),O\left(V\right)\right]$ tradeoff between these two objectives. Thus, we can balance the system's weighted sum power consumption and execution delay by adjusting $V$: For delay-sensitive applications, we can use a small value of $V$; while for energy-sensitive networks and delay-tolerant applications, a large value of $V$ can be adopted.
\label{rmktradeoff}
\end{rmk}

\section{Simulation Results}
In simulations, we assume $N$ mobile devices are located at an equal distance of $150$ m from the MEC server. {The small-scale fading channel power gains are exponentially distributed with unit mean, i.e., $\gamma_{i}\left(t\right)\sim \text{Exp} \left(\text{1}\right),i\in\mathcal{N}$.} Besides, $\tau=1$ ms, $\omega=10$ MHz, $N_{0}=-174$ dBm\slash Hz, $g_{0}=-40$ dB, $d_{0}=1$ m, $\theta=4$, $w_{i}= 1$, $\kappa_{{\rm{mob}},i}= 10^{-27}$, $f_{i,\max}=1$ GHz, $p_{i,\max}=500$ mW, $A_{i}\left(t\right)$ is uniformly distributed within $\left[0,A_{i,\max}\right]$, and $L_{i}=737.5$ cycles\slash bit, $i\in\mathcal{N}$ \cite{Miettinen10}. In addition, we set $f_{C_{m},\max}= 2.5$ GHz, $\kappa_{{\rm{ser}},m}=10^{-27}$, $m\in\mathcal{M}$. The simulation results are averaged over $10000$ time slots.
\vspace{-15pt}
\begin{figure}[ht]
\centering
\includegraphics[width=0.6\textwidth]{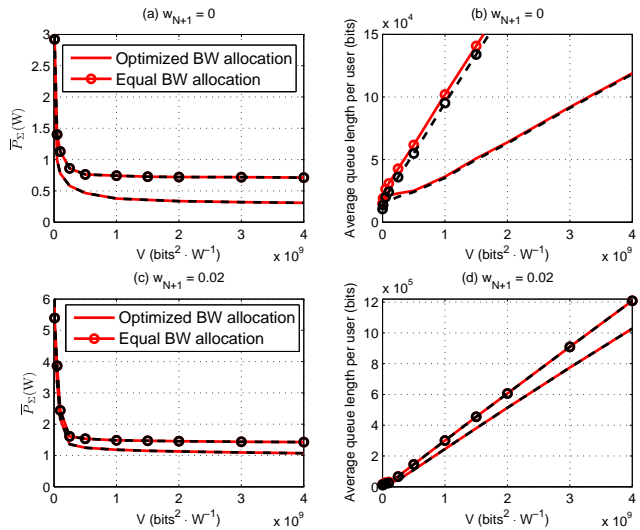}
\vspace{-20pt}
\caption{Average weighted sum power consumption/average sum queue length of the task buffers per mobile device vs. $V$, $N=5$, $M=8$, and $A_{i,\max}=8$ kbits. (The solid curves and dash curves correspond to Algorithm \ref{Algframework} and Algorithm \ref{Delayenhanced}, respectively, while the curves marked with circles are obtained by setting $\alpha_{i}\left(t\right)=1\slash N,i\in\mathcal{N}$ in the proposed algorithms.)}
\label{FIGPwWRQV}
\end{figure}

We first validate the theoretical results derived in Theorem \ref{performanalysis} for the proposed algorithms in Fig. \ref{FIGPwWRQV}. Two scenarios with $w_{N+1}=0$ and $0.02$, respectively, are considered. The average weighted sum power consumption for the scenario with $w_{N+1}=0$ ($w_{N+1}=0.02$) is shown in Fig. \ref{FIGPwWRQV}a) (Fig. \ref{FIGPwWRQV}c)), while the average sum queue length of the task buffers per mobile device, i.e., $\sum_{i\in\mathcal{N}}\overline{q}_{\Sigma,i}\slash N$, for $w_{N+1}=0$ ($w_{N+1}=0.02$) is shown in Fig. \ref{FIGPwWRQV}b) (Fig. \ref{FIGPwWRQV}d)). It can be observed from Fig. \ref{FIGPwWRQV}a) and Fig. \ref{FIGPwWRQV}c) that the average weighted sum power consumption decreases inversely proportional to the control parameter $V$ and converges to $P_{\Sigma,\mathbf{P_{2}}}^{\rm{opt}}$ when $V$ is sufficiently large. Meanwhile, as shown in Fig. \ref{FIGPwWRQV}b) and Fig. \ref{FIGPwWRQV}d), the average sum queue length of the task buffers increases linearly with $V$ for both Algorithm \ref{Algframework} and its delay-improved version (Algorithm \ref{Delayenhanced}), which verifies the $\left[O\left(1\slash V\right),O\left(V\right)\right]$ tradeoff between the average weighted sum power consumption and the average execution delay as demonstrated in Theorem \ref{performanalysis} and Remark \ref{rmktradeoff}. Furthermore, we see that increasing $w_{N+1}$ results in a higher average sum queue length, which comes from the slowdown of the server CPU frequencies.

The proposed algorithm and its delay-improved version with equal bandwidth allocation, i.e., $\alpha_{i}\left(t\right)=1\slash N,i\in\mathcal{N}$, are also evaluated in Fig. \ref{FIGPwWRQV}. We see the performance of both the average weighted sum power consumption and the average sum queue length of the task buffers are deteriorated compared to those achieved with optimized bandwidth allocation. This demonstrates the importance of a joint consideration on the radio and computational resource management for multi-user MEC systems. Besides, it is shown that the average weighted sum power consumptions under both Algorithm \ref{Algframework}  and Algorithm \ref{Delayenhanced} are the same, while the average sum queue length of the task buffers under the delay-improved mechanism is reduced, which agrees with Proposition \ref{propmodialgo}. Such queue length (i.e., execution delay) reduction is more obvious when $w_{N+1}$ and $V$ get closer to zero. This is because under these circumstances, the frequencies of the server CPU cores are much higher than what are needed for serving the mobile device with the highest value of $T_{i}\left(t\right)L_{i}^{-1}$. In other words, more computational resource provided by the MEC server is wasted under Algorithm \ref{Algframework} and can be re-allocated for other mobile devices. We only focus on Algorithm \ref{Delayenhanced} in the following, as it always achieves improved performance compared to Algorithm \ref{Algframework}.

\begin{figure}
  \begin{minipage}[t]{8cm}
    \centering
    \includegraphics[width=1\textwidth]{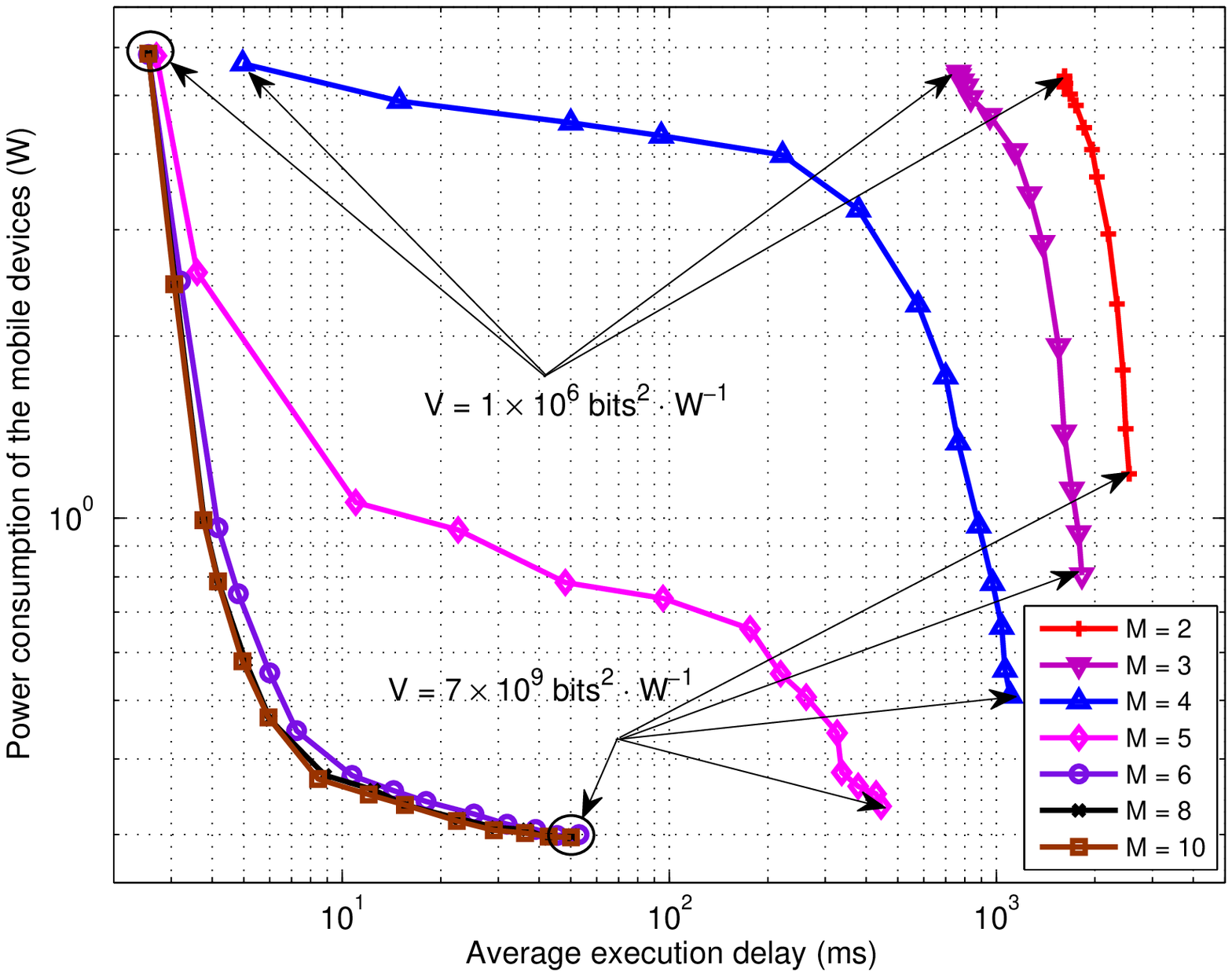}
    \vspace{-40pt}
    \caption{Power consumption of the mobile devices vs. average execution delay, $w_{N+1}=0$, $N=5$, and $A_{i,\max}=8$ kbits.}
    \label{ImpactMECcapability}
  \end{minipage}%
  \hspace{8pt}
  \begin{minipage}[t]{8cm}
    \centering
    \includegraphics[width=1\textwidth]{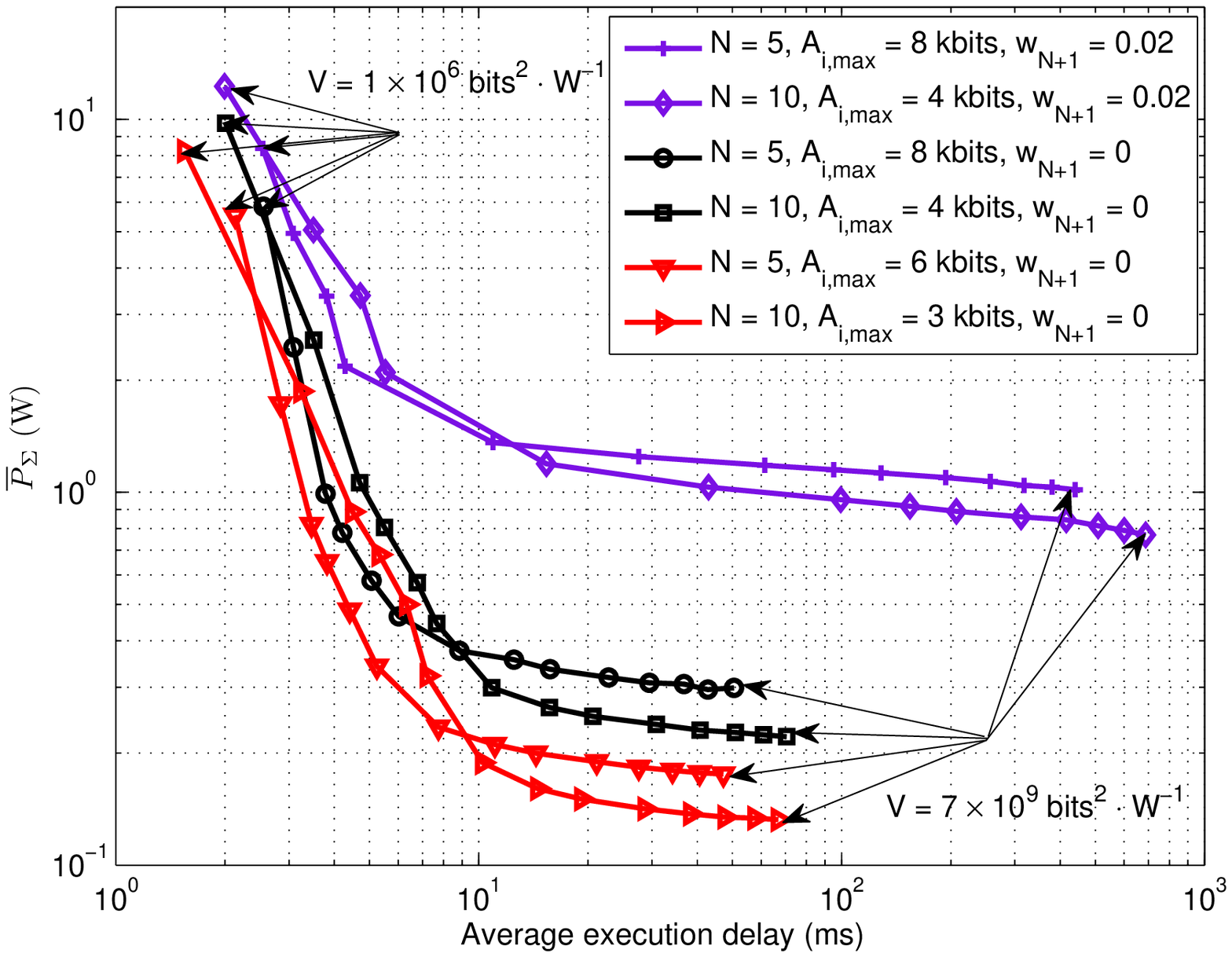}
      \vspace{-40pt}
    \caption{Weighted sum power consumption of the MEC system vs. average execution delay, $M=8$.}
    \label{MUdiversity}
  \end{minipage}
\end{figure}

In Fig. \ref{ImpactMECcapability}, we reveal the impacts of the computation capability of the MEC server by setting $w_{N+1}=0$ and varying the number of CPU cores $M$. In these scenarios, the power consumption of the MEC server is not addressed in the optimization and thus the server CPU cores will operate at their maximum frequencies. As can be seen from Fig. \ref{ImpactMECcapability}, the system performance improves as the value of $M$ increases, i.e., the power consumption of the mobile devices decreases for a given average execution delay.\footnote{The average execution delay is calculated by $\sum_{i\in\mathcal{N}}\overline{q}_{\Sigma,i}\slash \sum_{i\in\mathcal{N}}\lambda_{i}$ (time slots) according to Little's Law.} Also, by increasing $V$ from $1\times10^{6}$ to $7\times10^{9}\ {\rm{bit}}^{2}\cdot {\rm{W}}^{-1}$, the power consumption decreases significantly for all curves. Nonetheless, the behaviors of the curves for different values of $M$ are substantially different. For a relatively large value of $M$, i.e., $M=$ 6, 8 and 10, the average execution delay decreases sharply from 52 to 2.6 ms as $V$ decreases, and the power and delay performance follows the $\left[O\left(1\slash V\right),O\left(V\right)\right]$ tradeoff. However, for a relatively small value of $M$, the average execution delay has minor changes as $V$ is adjusted, i.e., around 2000 and 1000 ms for $M=$ 2 and 3, respectively, and the speed of power reduction increases as the average execution delay increases. This is because with insufficient computational resource at the MEC server, the task buffers cannot be stabilized even with $V=1\times 10^{6}\ {\rm{bit}}^{2}\cdot {\rm{W}}^{-1}$ within 10000 time slots. Similar phenomenon exists for the curves with a medium value of $M$ (such as 4 and 5) when a large value of $V$ is used. In addition, the performance gain achieved by increasing $M$ from 6 to 10 is negligible, which indicates that the computation capability of the MEC server should be chosen carefully to balance the system cost and the achievable performance. In particular, once the system is constrained by the radio resource, there is no need to deploy too much computational resource.



By varying $A_{i,\max}$ and $N$, we show the relationship between the average weighted sum power consumption and the average execution delay in Fig. \ref{MUdiversity}. Similar to Fig. \ref{FIGPwWRQV} and Fig. \ref{ImpactMECcapability}, the average execution delay increases as the weighted sum power consumption decreases, which indicates that a proper $V$ should be chosen to balance the two desirable objectives. For instance, with $N=5$, $A_{i,\max}=8$ kbits and $w_{N+1}=0$, if the average execution delay requirement is 20 ms, $V=3\times 10^{9}$ ${\rm{bits}}^{2}\cdot {\rm{W}}^{-1}$ can be chosen, and the sum power consumption of the mobile devices will be about 0.3 W. Besides, with a given execution delay, the average weighted sum power consumption increases with the computation task arrival rate, which fits the intuition, i.e., as the workload of the MEC system becomes more intensive, more power is needed for keeping the task buffers stable. In addition, given a total task arrival rate in the MEC system, increasing the number of mobile devices while decreasing the task arrival rate at individual device results in a lower weighted sum power consumption when $V$ is tuned sufficiently large, i.e., the proposed algorithm is approaching the optimal performance. This is due to the increased multi-user diversity gain and the availability of extra local processing units.


\begin{figure}
  \begin{minipage}[t]{8cm}
    \centering
    \includegraphics[width=1\textwidth]{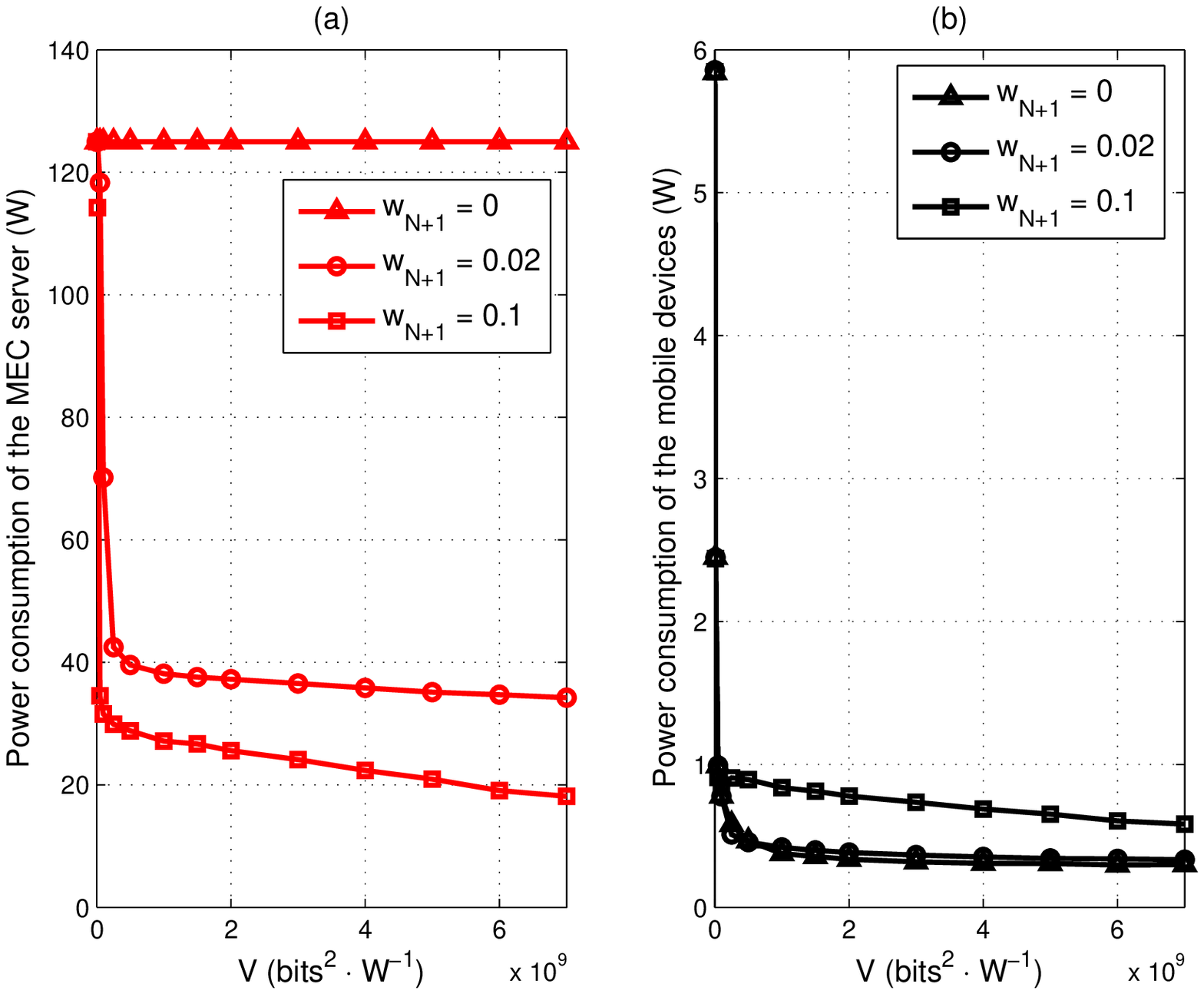}
       \vspace{-40pt}
    \caption{Power consumption of the MEC server/mobile devices vs. $V$, $M=8$, $N=5$ and $A_{i,\max}=8$ kbits.}
    \label{ImpactWeightPower}
  \end{minipage}%
  \hspace{8pt}
  \begin{minipage}[t]{8cm}
    \centering
    \includegraphics[width=1\textwidth]{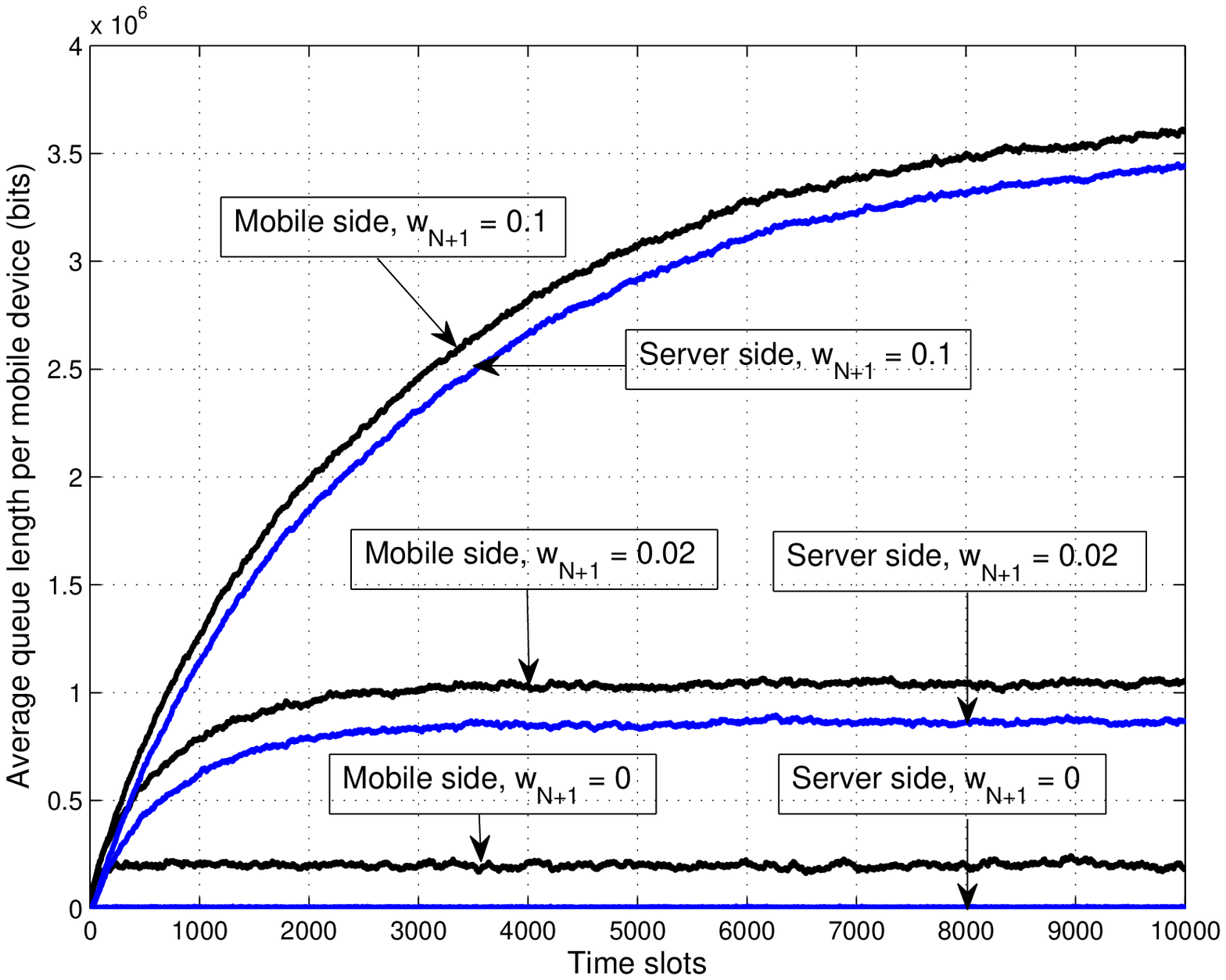}
       \vspace{-40pt}
    \caption{Average queue length per mobile device vs. time, $N=5$, $M=8$, $A_{i,\max}=8$ kbits, and $V=7\times 10^{9}\ {{\rm{bits}}^{2}\cdot \mathrm{W}^{-1}}$.}
    \label{MonitorQueue}
  \end{minipage}
\end{figure}

We show the impacts of $w_{N+1}$ on the power consumption of the MEC server and the mobile devices in Fig. \ref{ImpactWeightPower}. {It can be observed that by increasing $w_{N+1}$, the power consumption of the MEC server decreases while that of the mobile devices increases for a given value of $V$, which confirms that the power consumption of the mobile devices and the MEC server can be balanced by adjusting the weighting factors.} To be more specific, the power consumption of the MEC server is a constant for $w_{N+1}=0$ due to the optimal frequencies of the server CPU cores derived in (\ref{mcorefreq}). By increasing $w_{N+1}$ from 0 to 0.02, the power consumption of the MEC server is greatly reduced while that of the mobile devices has only minor increase, which leads to higher queue lengths at the task buffers. By increasing $w_{N+1}$ from 0.02 to 0.1, the power consumption of the MEC server further decreases and that of the mobile devices increases noticeably. In this case, the local CPUs needs to increase their CPU-cycle frequencies to compensate the processing speed reduction at the MEC server in order to keep the task buffers stable. For better demonstration, we show the evolution of the average queue lengths at the mobile devices (i.e., $\sum_{i\in\mathcal{N}}Q_{i}\left(t\right)\slash N$) and the MEC server (i.e., $\sum_{i\in\mathcal{N}}T_{i}\left(t\right)\slash N$) over time in Fig. \ref{MonitorQueue}. We see for $w_{N+1}=0$, the amount of waiting tasks in the server side buffers is near zero as the MEC server is able to execute all tasks offloaded from the mobile devices, and the queue length of the local task buffers is also maintained at a small level as most tasks can be offloaded to and executed by the MEC server. For $w_{N+1}=0.02$, the amount of tasks in the buffers increases due to the decrease of CPU frequencies at the MEC server. With $w_{N+1}=0.1$, the queue lengths of the task buffers at both the mobile devices and the MEC server are much higher than those achieved by $w_{N+1}=0.02$ and keep increasing as time elapses, which implies that the system cannot converge to its stabilized performance within the time window used in simulations, and in turn lifts up the local CPU-cycle frequencies and the power consumption of the mobile devices.

\begin{figure}[h]
\centering
\includegraphics[width=0.6\textwidth]{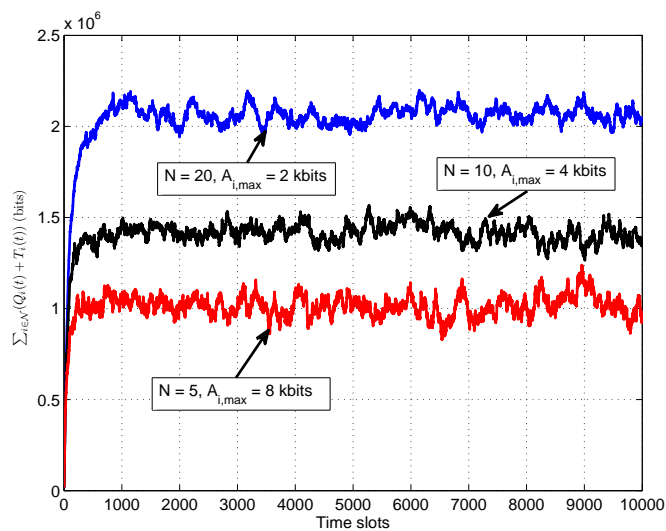}
\vspace{-20pt}
\caption{Sum queue length of the task buffers vs. time slots, $M=8$, $\sum_{i\in\mathcal{N}}\lambda_{i}=20$\ kbits, $V = 7\times 10^{9}\ {\rm{bit}}^{2}\cdot {\rm{W}}^{-1}$.}
\label{SumQueLengthTS}
\end{figure}

{Finally, we demonstrate the impact of the network size on the convergence time under the proposed algorithm by tracing the sum queue length of the task buffers for systems with different numbers of mobile devices, while keeping the total computation task arrival rate in the MEC system unchanged at $\sum_{i\in\mathcal{N}}\lambda_{i}= 20\ \text{kbits}$. It is observed from Fig. \ref{SumQueLengthTS} that for all the three cases, the sum queue length of the task buffers increases at the beginning, and stabilizes at around $1\times 10^6\ {\text{bits}}$, $1.4 \times 10^{6}\ \text{bits}$, and $2 \times 10^{6}\ \text{bits}$ for $N=5$, $10$, and $20$, respectively. Besides, we see that although a larger value of $N$ leads to a longer convergence time, the differences are minor and the task buffers are able to be stabilized within approximately 1000 time slots for all the cases, which indicates that the proposed algorithms are scalable to MEC systems with a reasonable amount of mobile devices.}



\section{Conclusions}
In this paper, we investigated stochastic joint radio and computational resource management for multi-user mobile-edge computing (MEC) systems. A low-complexity online algorithm based on Lyapunov optimization was proposed, which achieves asymptotic optimality. Furthermore, a delay-improved mechanism was designed to reduce the execution delay. Performance analysis, as well as simulations, explicitly characterize the tradeoff between the average weighted sum power consumption and the average execution delay in a multi-user MEC system. Moreover, the impacts of various parameters were revealed, which confirm the importance of the joint radio and computational resource management, and provide valuable guidelines for the practical deployment of MEC systems. For future investigation, it would be interesting to extend the findings in this work to scenarios with fairness considerations among multiple devices and seek distributed implementation methodologies. {Other research problems, such as mobility-aware resource management, dynamic access control and user-server association, should also be investigated.}

\vspace{-15pt}
\appendix
\vspace{-5pt}
\subsection{Proof for Lemma \ref{Lyvdriftpenaltyboundlma}}
By squaring both sides of the local task buffer dynamics in (\ref{bufferdynamics}), we have
\begin{equation}
\begin{split}
Q^{2}_{i}\left(t+1\right)
&=\left(\max\{Q_{i}\left(t\right)-D_{\Sigma,i}\left(t\right),0\}\right)^{2}+A_{i}^{2}\left(t\right)+2A_{i}\left(t\right)\cdot \max\{Q_{i}\left(t\right)-D_{\Sigma,i}\left(t\right),0\}\\
&\leq \left(Q_{i}\left(t\right)-D_{\Sigma,i}\left(t\right)\right)^{2}+A_{i}^{2}\left(t\right)+2 A_{i}\left(t\right) Q_{i}\left(t\right)\\
&=Q_{i}^{2}\left(t\right)-2Q_{i}\left(t\right)\left(D_{\Sigma,i}\left(t\right)-A_{i}\left(t\right)\right)+A_{i}^{2}\left(t\right)+D_{\Sigma,i}^{2}\left(t\right).
\end{split}
\label{proofLemmaDPPbound1Q}
\end{equation}
By moving $Q_{i}^{2}\left(t\right)$ to the left-hand side of (\ref{proofLemmaDPPbound1Q}), dividing both sides by 2, and summing up the inequalities for $i=1,\cdots,N$, we have
\begin{equation}
\frac{1}{2}\sum_{i\in\mathcal{N}}\left[Q_{i}^{2}\left(t+1\right)-Q_{i}^{2}\left(t\right)\right]\leq -\sum_{i\in\mathcal{N}}Q_{i}\left(t\right)\left(D_{\Sigma,i}\left(t\right)-A_{i}\left(t\right)\right) + \sum_{i\in\mathcal{N}}\frac{A_{i}^{2}\left(t\right)+D_{\Sigma,i}^{2}\left(t\right)}{2}.
\label{proofLemmaDPPbound2Q}
\end{equation}

Similar manipulations can be made for the MEC server task buffer dynamics in (\ref{bufferdynamicsMEC}) as follows:
\begin{equation}
\begin{split}
T^{2}_{i}\left(t+1\right)
&\leq \left(\max\{T_{i}\left(t\right)-D_{s,i}\left(t\right),0\}+D_{r,i}\left(t\right)\right)^{2}\\
&=\left(\max\{T_{i}\left(t\right)-D_{s,i}\left(t\right),0\}\right)^{2}+D_{r,i}^{2}\left(t\right)+2D_{r,i}\left(t\right)\cdot \max\{T_{i}\left(t\right)-D_{s,i}\left(t\right),0\}\\
&\leq \left(T_{i}\left(t\right)-D_{s,i}\left(t\right)\right)^{2}+D_{r,i}^{2}\left(t\right)+2 D_{r,i}\left(t\right) T_{i}\left(t\right)\\
&=T_{i}^{2}\left(t\right)-2T_{i}\left(t\right)\left(D_{s,i}\left(t\right)-D_{r,i}\left(t\right)\right)+D_{r,i}^{2}\left(t\right)+D_{s,i}^{2}\left(t\right).
\end{split}
\label{proofLemmaDPPbound1T}
\end{equation}
By moving $T_{i}^{2}\left(t\right)$ to the left-hand side of (\ref{proofLemmaDPPbound1T}), dividing both sides by 2, and summing up the inequalities for $i=1,\cdots,N$, we have
\begin{equation}
\frac{1}{2}\sum_{i\in\mathcal{N}}\left[T_{i}^{2}\left(t+1\right)-T_{i}^{2}\left(t\right)\right]\leq -\sum_{i\in\mathcal{N}}T_{i}\left(t\right)\left(D_{s,i}\left(t\right)-D_{r,i}\left(t\right)\right) + \sum_{i\in\mathcal{N}}\frac{D_{r,i}^{2}\left(t\right)+D_{s,i}^{2}\left(t\right)}{2}.
\label{proofLemmaDPPbound2T}
\end{equation}
By summing up (\ref{proofLemmaDPPbound2Q}) and (\ref{proofLemmaDPPbound2T}), we have
\begin{equation}
\begin{split}
L\left(\bm{\Theta}\left(t+1\right)\right)-L\left(\bm{\Theta}\left(t\right)\right)&\leq -\sum_{i\in\mathcal{N}}\left[Q_{i}\left(t\right)\left(D_{\Sigma,i}\left(t\right)-A_{i}\left(t\right)\right)+T_{i}\left(t\right)\left(D_{s,i}\left(t\right)-D_{r,i}\left(t\right)\right)\right]\\ &+ \frac{1}{2}\sum_{i\in\mathcal{N}}\left(A_{i}^{2}\left(t\right)+D^{2}_{\Sigma,i}\left(t\right)
+D^{2}_{s,i}\left(t\right)+D_{r,i}^{2}\left(t\right)\right).
\end{split}
\label{DPPboundwithoutExp}
\end{equation}
Since $\log_{2}\left(1+x\right)\leq \frac{x}{\ln 2}$ and $\log^{2}_{2}\left(1+x\right)\leq \frac{2x}{\left(\ln 2\right)^{2}}$ for $x\geq 0$, we have $\mathbb{E}\left[D_{r,i}^{2}\left(t\right)|\bm{\Theta}\left(t\right)\right]\leq \frac{\omega \eta_{i}}{\ln 2}$ and $\mathbb{E}\left[D_{\Sigma,i}^{2}\left(t\right)|\bm{\Theta}\left(t\right)\right]\leq \mathbb{E}\left[\left(\tau f_{i,\max}L_{i}^{-1}+D_{r,i}\left(t\right)\right)^{2}|\bm{\Theta}\left(t\right)\right]\leq \tau^{2}f_{i,\max}^{2}L_{i}^{-2}+\eta_{i}\left(f_{i,\max}L_{i}^{-1}+\frac{\omega}{\ln 2}\right)$ with $\eta_{i}=\frac{2g_{0}\overline{\gamma_{i}}p_{i,\max}d_{0}^{\theta}\tau^{2}}{\ln 2 N_{0}d_{i}^{\theta}}$. Finally, by adding $V\cdot P_{\Sigma}\left(t\right)$ at both sides of (\ref{DPPboundwithoutExp}) and taking the expectation on both sides conditioned on $\bm{\Theta}\left(t\right)$, we can obtain the desired result in (\ref{Lyvdriftpenaltybound}), where
\begin{equation}
C=\frac{1}{2}\sum_{i\in\mathcal{N}}\left[A_{i,\max}^{2}+\left(\sum_{m\in\mathcal{M}}f_{C_{m},\max}\tau L_{i}^{-1}\right)^{2}+\left(f_{i,\max}\tau L_{i}^{-1}\right)^{2}+\eta_{i}\left(f_{i,\max}L_{i}^{-1}+\frac{2\omega }{\ln 2}\right)\right].
\end{equation}

\vspace{-20pt}
\subsection{Proof for Lemma \ref{propertySP3}}
Since $T_{i}\left(t\right)\geq 0, i\in\mathcal{N}$, there exists an optimal solution such that $\sum_{n\in\mathcal{N}}D_{s,n}\left(t\right)L_{n}= \sum_{m\in\mathcal{M}}f_{C,m}\left(t\right)\tau$. For $\sum_{m\in\mathcal{M}} f_{C,m}\left(t\right)=0$, the optimal scheduling decision $\mathbf{D}_{s}\left(t\right)=\mathbf{0}$. For $\sum_{m\in\mathcal{M}} f_{C,m}\left(t\right)>0$, suppose there is an optimal scheduling decision $\mathbf{D}_{s}\left(t\right)$ with $|\mathbf{D}_{s}\left(t\right)|_{0}\geq 2$. Denote $\mathcal{S}\triangleq \{n|D_{s,n}\left(t\right)>0,n\in\mathcal{N}\}$ where $|\mathcal{S}|=|\mathbf{D}_{s}\left(t\right)|_{0}$. We can construct a new feasible scheduling decision $\hat{\mathbf{D}}_{s}\left(t\right)$ with $\hat{D}_{s,i^{\max}_{\mathcal{S}}}\left(t\right)= D_{s,i^{\max}_{\mathcal{S}}}\left(t\right)+L_{i^{\max}_{\mathcal{S}}}^{-1}\sum_{j\in\mathcal{S},j\neq i^{\max}_{\mathcal{S}}}D_{s,j}\left(t\right)L_{j}$ and $\hat{D}_{s,i}\left(t\right)=0, i\neq i^{\max}_{\mathcal{S}}$, where $i^{\max}_{\mathcal{S}}$ is the mobile device in $\mathcal{S}$ with the highest value of $T_{i}\left(t\right)L_{i}^{-1}$. With the constructed solution, the value of the objective function will decrease by
\begin{equation}
-\sum_{i\in\mathcal{N}}T_{i}\left(t\right)\left(D_{s,i}\left(t\right)-\hat{D}_{s,i}\left(t\right)\right)
=\sum_{j\in\mathcal{S},j\neq i_{\mathcal{S}}^{\max}}D_{s,j}\left(t\right)L_{i_{\mathcal{S}}^{\max}}^{-1}\left(T_{i_{\mathcal{S}}^{\max}}\left(t\right)L_{j}-T_{j}\left(t\right)L_{i_{\mathcal{S}}^{\max}}\right)\geq 0,
\end{equation}
i.e., $\hat{\mathbf{D}}_{s}\left(t\right)$ with $|\hat{\mathbf{D}}_{s}\left(t\right)|_{0}= 1$ performs no worse than $\mathbf{D}_{s}\left(t\right)$. In addition, if $i_{\mathcal{S}}^{\max}\neq i_{\mathcal{N}}^{\max}$, we can serve the $i_{\mathcal{N}}^{\max}$th device with the CPU cycles that are originally allocated for the $i_{\mathcal{S}}^{\max}$th device, and the value of the objective function will further decrease by $\hat{D}_{s,i_{\mathcal{S}}^{\max}}\left(t\right)L^{-1}_{i_{\mathcal{N}}^{\max}}(T_{i_{\mathcal{N}}^{\max}}\left(t\right)L_{i^{\max}_\mathcal{S}}-
T_{i_{\mathcal{S}}^{\max}}\left(t\right)L_{i_{\mathcal{N}}^{\max}})\geq 0$. In other words, there exists an optimal solution for $\mathbf{SP_{3}}$ for a given $\mathbf{f}_{C}\left(t\right)$ such that $D_{s,i_{\mathcal{N}}^{\max}}\left(t\right)\geq 0$ and $D_{s,i}\left(t\right)=0, i\neq i_{\mathcal{N}}^{\max}$.


\subsection{Proof for Theorem \ref{performanalysis}}
Denote the optimal solution for the per-time slot problem in time slot $t$ as $\mathbf{X}^{\star}\left(t\right)$. Since the system operation at each time slot is the optimal solution of the per-time slot problem, we have
\begin{equation}
\begin{split}
\Delta_{V}\left(\bm{\Theta}\left(t\right)\right)&\leq C - \mathbb{E}\left[\sum_{i\in\mathcal{N}}Q_{i}\left(t\right)\left(D^{\star}_{\Sigma,i}\left(t\right)-A_{i}\left(t\right)\right)|\bm{\Theta}\left(t\right)\right]\\
&-\mathbb{E}\left[\sum_{i\in\mathcal{N}}T_{i}\left(t\right)\left(D^{\star}_{s,i}\left(t\right)-D^{\star}_{r,i}\left(t\right)\right)|\bm{\Theta}\left(t\right)\right]+
V\cdot \mathbb{E}\left[P^{\star}_{\Sigma}\left(t\right)|\bm{\Theta}\left(t\right)\right]\\
&\overset{\left(a\right)}{\leq} C + \sum_{i\in\mathcal{N}}\left(Q_{i}\left(t\right)+T_{i}\left(t\right)\right)\delta + V \cdot \left(P_{\Sigma,\mathbf{P_{3}}}^{\rm{opt}}+\delta\right),\\
\end{split}
\end{equation}
where (\emph{a}) holds since the stationary and randomized policy $\Pi$ is sub-optimal for $\mathbf{P}_{\rm{PTS}}$ and also due to the properties of $\Pi$ as shown in Lemma \ref{stationaryrandomizedproperty}. By taking
$\delta\rightarrow 0$, we have
\begin{equation}
\Delta_{V}\left(\bm{\Theta}\left(t\right)\right)\leq C + V\cdot P_{\Sigma,\mathbf{P_{3}}}^{\rm{opt}} = C + V\cdot P_{\Sigma,\mathbf{P_{2}}}^{\rm{opt}}.
\label{prfThm12}
\end{equation}
By taking the expectation on both sides of (\ref{prfThm12}), summing up the inequalities for $t=0,\cdots,T-1$, we have
\begin{equation}
\mathbb{E}\left[L\left(\bm{\Theta}\left(T\right)\right)\right]-\mathbb{E}\left[L\left(\bm{\Theta}\left(0\right)\right)\right]
+V\cdot \sum_{t=0}^{T-1}\mathbb{E}\left[P^{\star}_{\Sigma}\left(t\right)\right]\leq
T\cdot \left(C + V\cdot P_{\Sigma,\mathbf{P_{2}}}^{\rm{opt}}\right),
\label{prfThm13}
\end{equation}
where $P^{\star}_{\Sigma}\left(t\right)$ is the weighted sum power consumption in time slot $t$ under $\mathbf{X}^{\star}\left(t\right)$. Since $\bm{\Theta}\left(0\right)=\mathbf{0}$ and $\mathbb{E}\left[L\left(\bm{\Theta}\left(T\right)\right)\right]\geq 0$, by dividing both sides of (\ref{prfThm13}) by $T$ and letting $T$ go to infinite, we have $\overline{P}_{\Sigma}^{\star}\leq P_{\Sigma,\mathbf{P_{2}}}^{\rm{opt}} + C\cdot V^{-1}$.

We now proceed to show the mean rate stability of the task buffers. Denote $\Xi\left(T\right)\triangleq T \cdot \left(C + V \cdot P_{\Sigma,\mathbf{P_{2}}}^{\rm{opt}}\right)- V\cdot \sum_{t=0}^{T-1}\mathbb{E}\left[P_{\Sigma}^{\star}\left(t\right)\right]+\mathbb{E}\left[L\left(\bm{\Theta}\left(0\right)\right)\right]$.
Based on (\ref{prfThm13}), we have $\mathbb{E}\left[L\left(\bm{\Theta}\left(T\right)\right)\right]\leq \Xi\left(T\right)$, i.e., $\mathbb{E}\left[Q^{2}_{i}\left(T\right)\right]\leq 2 \Xi\left(T\right)$ and $\mathbb{E}\left[T^{2}_{i}\left(T\right)\right]\leq 2 \Xi\left(T\right)$. Thus, we have $0\leq \mathbb{E}\left[|Q_{i}\left(T\right)|\right]\leq \sqrt{2 \Xi\left(T\right)}$ and $0\leq \mathbb{E}\left[|T_{i}\left(T\right)|\right]\leq \sqrt{2 \Xi\left(T\right)},i\in\mathcal{N}$. Since $\lim_{T\rightarrow +\infty}\frac{\sqrt{2\Xi\left(T\right)}}{T}=0$, we have
\begin{equation}
\lim_{T\rightarrow +\infty} \frac{\mathbb{E}\left[|Q_{i}\left(T\right)|\right]}{T}=\lim_{T\rightarrow +\infty} \frac{\mathbb{E}\left[|T_{i}\left(T\right)|\right]}{T}=0,i\in\mathcal{N},\\
\end{equation}
i.e., $\forall i\in\mathcal{N}$, $Q_{i}\left(t\right)$ and $T_{i}\left(t\right)$ are mean rate stable under Algorithm \ref{Algframework}.

In order to show the upper bound of the average sum queue length, we plug the stationary and randomized policy $\tilde{\Pi}$ into the right-hand side of (\ref{Lyvdriftpenaltybound}), i.e.,
\begin{equation}
\Delta_{V}\left(\bm{\Theta}\left(t\right)\right)\leq C - \epsilon\cdot \sum_{i\in\mathcal{N}}\left(Q_{i}\left(t\right)+T_{i}\left(t\right)\right)+ V \cdot \Psi\left(\epsilon\right).
\end{equation}
By taking the expectation on both sides of the inequality, summing up the inequalities for $t=0,\cdots,T-1$, dividing both sides by $T \epsilon$, and letting $T$ go to infinite, we have
\begin{equation}
\begin{split}
\sum_{i\in\mathcal{N}}\overline{q}_{\Sigma,i}
\leq \frac{C+ V\cdot \left(\Psi\left(\epsilon\right)-\lim\limits_{T\rightarrow +\infty}\frac{1}{T}\sum_{t=0}^{T-1}\mathbb{E}\left[P^{\star}_{\Sigma}\left(t\right)\right]\right)}{\epsilon} \leq \frac{C+ V\cdot \left(\Psi\left(\epsilon\right)-P_{\Sigma,\mathbf{P_{2}}}^{\rm{opt}}\right)}{\epsilon}.
\end{split}
\end{equation}
With the aid of Proposition \ref{propmodialgo}, we can show Theorem \ref{performanalysis} also holds for Algorithm \ref{Delayenhanced}.

\end{document}